\documentclass[11pt]{article}

\usepackage[english]{babel}

\usepackage[margin=1in]{geometry}

\usepackage{graphicx}
\usepackage{amsfonts,amsmath,amssymb,amsthm} 
\usepackage{comment}
\usepackage{xcolor}
\definecolor{ForestGreen}{rgb}{0.1333,0.5451,0.1333}
\definecolor{DarkRed}{rgb}{0.80,0,0}
\definecolor{Red}{rgb}{1,0,0}
\usepackage[linktocpage=true,
pagebackref=true,colorlinks,
linkcolor=DarkRed,citecolor=ForestGreen,
bookmarks,bookmarksopen,bookmarksnumbered]
{hyperref}
\usepackage{cleveref}
\usepackage{algorithm}
\usepackage[noend]{algpseudocode}
\usepackage{booktabs}
\usepackage{multirow}

\theoremstyle{plain}
\newtheorem{theorem}{Theorem}[section]
\newtheorem{lemma}[theorem]{Lemma}

\newtheorem{corollary}[theorem]{Corollary}

\theoremstyle{definition}
\newtheorem{definition}[theorem]{Definition}

\newcommand{\R}{\mathbb{R}}

\newcommand{\E}{\mathbb{E}}

\newcommand{\bdf}{\boldsymbol{f}}

\newcommand{\nb}[1]{{\color{black} #1}}
\newcommand{\nbr}[1]{{\color{black} #1}}

\title{The Incentive Guarantees Behind Nash Welfare in Divisible Resources Allocation\thanks{A preliminary version of this paper is published in WINE'23. New results are added in this version, and are presented in \Cref{sec:2-agents}.}}

\author{
    Xiaohui Bei \\
    Nanyang Technological University \\
    \texttt{xhbei@ntu.edu.sg} \\
    \and
    Biaoshuai Tao \\
    Shanghai Jiao Tong University \\
    \texttt{bstao@sjtu.edu.cn} \\
    \and
    Jiajun Wu \\
    Huawei Technologies \\
    \texttt{wujiajun17@huawei.com} \\
    \and
    Mingwei Yang \\
    Stanford University \\
    \texttt{mwyang@stanford.edu}
}
\date{}

\begin{document}
\maketitle
\date{}

\begin{abstract}
    We study the problem of allocating divisible resources among $n$ agents, hopefully in a fair and efficient manner.
    With the presence of strategic agents, additional incentive guarantees are also necessary, and the problem of designing fair and efficient mechanisms becomes much less tractable.
    While there are flourishing positive results against strategic agents for homogeneous divisible items, very few of them are known to hold in cake cutting.

    We show that the Maximum Nash Welfare (MNW) mechanism, which provides desirable fairness and efficiency guarantees and achieves an \emph{incentive ratio} of $2$ for homogeneous divisible items, also has an incentive ratio of $2$ in cake cutting.
    Remarkably, this result holds even without the free disposal assumption, which is hard to get rid of in the design of truthful cake cutting mechanisms.

    Moreover, we show that, for cake cutting, the Partial Allocation (PA) mechanism proposed by Cole \textit{et al.} (EC'13), which is truthful and $1/e$-MNW for homogeneous divisible items, has an incentive ratio between $[e^{1 / e}, e]$ and when randomization is allowed, can be turned to be truthful in expectation.
    Given two alternatives for a trade-off between incentive ratio and Nash welfare provided by the MNW and PA mechanisms, we establish an interpolation between them for both cake cutting and homogeneous divisible items.

    Finally, we study the optimal incentive ratio achievable by envy-free cake cutting mechanisms.
    We first give an envy-free mechanism for two agents with an incentive ratio of $4 / 3$.
    Then, we show that any envy-free cake cutting mechanism with the connected pieces constraint has an incentive ratio of $\Theta(n)$.
\end{abstract}
\section{Introduction}

In the resource allocation problem, the goal is to allocate a set of resources among $n$ competing agents, and agents may express varied preferences over the resources.
In this paper, we focus on resources that are \emph{infinitely divisible}, with examples including cake and land.
When agents' strategic behaviors are not taken into account, i.e., agents truthfully report their preferences over \nbr{the resources}, the problem is often referred to as \emph{cake cutting}\footnote{\nb{While a majority of cake cutting literature treats fairness as the primary objective, it is also common to refer cake cutting as a general resource allocation problem \cite{DBLP:conf/faw/BuS23,DBLP:conf/atal/AumannDH13}.}}, where cake is a metaphor for a heterogeneous divisible item and is usually represented by interval $[0, 1]$.
The cake-cutting problem has been extensively studied for many decades by mathematicians, economists, political scientists, and more recently, computer scientists from various perspectives \cite{DBLP:books/daglib/0017730,DBLP:books/daglib/0017734,DBLP:reference/choice/Procaccia16,DBLP:books/daglib/0017738}.

A popular criterion for defining a ``good'' allocation is fairness, and the most prominent fairness notions include \emph{envy-freeness} and \emph{proportionality}.
Specifically, an allocation is called envy-free if each agent likes his received bundle at least as much as every other agent's bundle, and proportional if each agent receives a bundle with a value of at least $1 / n$ of the value he assigns to the entire cake.
It is well known that \nb{with mild assumptions}, an envy-free allocation always exists \cite{brams1995envy}, even if the bundle received by each agent is required to be a connected piece \cite{edward1999rental}.
Another important measure of an allocation is that of efficiency, which evaluates the solution from a global point of view.

The Maximum Nash Welfare (MNW) rule, which returns an allocation that maximizes Nash welfare, i.e., the geometric mean of agents' utilities, has been shown to be promising as it possesses various desirable properties including envy-freeness, proportionality and \emph{Pareto optimality}, an appealing efficiency concept \cite{weller1985fair}.
Moreover, \cite{berliant1992fair} strengthens the fairness guarantee provided by the MNW rule by showing that it satisfies \emph{group envy-freeness}, and \cite{segal2019monotonicity} shows that the MNW rule has intuitive resource and population monotonicity properties\footnote{\nb{Resouce (resp. population) monotonicity requires that when the set of resources (resp. agents) grows or shrinks, and the resources are re-divided using the same rule, the utility of each remaining agent must change in the same direction.}}.
\nbr{To make the discussion of the computational efficiency of the MNW rule meaningful, it is widely assumed that agents' value density functions are \emph{piecewise constant}, which admit succinct representations and can be used to approximate most natural real functions within an arbitrary precision.
Under the assumption that agents' value density functions are piecewise constant, the MNW allocations exactly coincide with the \emph{competitive equilibrium from equal incomes (CEEI)} allocations \cite[Volume 2, Chapter 14]{arrow1981handbook}, which in turn can be computed in polynomial time via the convex program of Eisenberg and Gale \cite{eisenberg1959consensus}.}

However, a challenging issue arises with self-interested agents, who may misreport their valuations in the hope of a better allocation for their own.
When strategic behaviors are involved, agents have varied incentives to truthfully report their preferences with homogeneous divisible items and with a heterogeneous divisible item, i.e., a cake.
When agents' value density functions are piecewise constant, it seems that we can reduce a cake to multiple homogeneous divisible items by partitioning the cake into multiple pieces such that each agent's valuation is uniform on each piece.
However, this interpretation is inaccurate when strategic behaviors are concerned: agents could possibly misreport the breakpoint positions in their valuations, and therefore modify the boundaries between items, or even add new items or merge items.
\nb{To illustrate, assume that there are two agents and a cake $[0, 1]$, where agent $1$ has different values for $[0, 1/2]$ and $[1/2, 1]$ and agent $2$ has the same value for the entire cake.
When both agents report truthfully, the above reduction would treat the cake as two homogeneous divisible items $[0, 1 / 2]$ and $[1 / 2, 1]$.
However, if agent $1$ maliciously claims that his preference is uniform over the cake, then the reduction would treat the entire cake as one homogeneous divisible item, resulting in a merge of items. In the case where agent $1$ knows that the mechanism breaks the tie by allocating the left-most piece to her and she has a higher value on $[0,1/2]$, this manipulation is beneficial for agent $1$.}
Therefore, there are less incentives for agents to truthfully report with a cake than with homogeneous divisible items.

The existence of strategic behaviors motivates the study of \emph{truthful} mechanisms, requiring that each agent cannot benefit from manipulation in \emph{any} case.
A strong negative result exists for cake cutting that there does not exist a deterministic, truthful, and proportional cake cutting mechanism when agents' valuations are piecewise constant \cite{DBLP:journals/ai/BuST23}.
On the positive side, \cite{DBLP:journals/geb/ChenLPP13} gives the first truthful envy-free cake cutting mechanism that works when agents’ valuations are \emph{piecewise uniform}\footnote{\nbr{A function over $[0, 1]$ is piecewise uniform if it is piecewise constant and its value at every point is either $0$ or $1$.}}, which is essentially equivalent to the MNW mechanism and even satisfies the stronger notion of \emph{group strategyproofness} \cite{DBLP:conf/wine/AzizY14}, and the design of truthful and fair mechanisms has also been studied for valuations more restrictive than piecewise uniform \cite{DBLP:conf/aaai/AlijaniFGST17,DBLP:journals/algorithmica/SeddighinFGAT19,DBLP:conf/isaac/AsanoU20}.
For homogeneous divisible items, dividing each item evenly among all the agents yields a truthful and envy-free mechanism.
However, this mechanism obviously has a poor performance in efficiency.
On the other hand, \cite{DBLP:conf/sigecom/ColeGG13} presents a truthful mechanism that provides a $1/e$-approximation guarantee for maximum Nash welfare ($1/e$-MNW for short), although the mechanism leaves some parts of the items unallocated.

Given all the difficulties to incorporate truthfulness for the divisible resources allocation problem, and since it is fairness and efficiency that serve as our primary concern, it should be satisfactory for us to relax truthfulness.
To further justify, even facing a non-truthful mechanism, it is also costly for agents to misreport as computing the best response of various mechanisms seems to be non-trivial and the information about other agents, which tends to be hard to elicit in practice, is also necessary.
Along this line, various relaxations of truthfulness are proposed for cake cutting \cite{DBLP:journals/ai/BuST23,DBLP:conf/dagstuhl/BramsJK07a,DBLP:journals/scw/OrtegaS22}.

Another popular notion to quantitatively measure the degree of untruthfulness is \emph{incentive ratio}, which has been considered on the literature in Fisher markets \cite{DBLP:journals/iandc/ChenDTZZ22,DBLP:conf/esa/ChenDZ11,DBLP:conf/icalp/ChenDZZ12}, resource sharing \cite{DBLP:conf/ipps/ChengDL20,DBLP:conf/sigecom/ChengDLY22,DBLP:journals/dam/ChenCDQY19}, housing markets \cite{todo2019analysis}, and resource allocation \cite{DBLP:journals/jcss/HuangWWZ24, DBLP:conf/aaai/XiaoL20,tao2023fair,DBLP:conf/atal/0037SX24}.
Informally, the incentive ratio of a mechanism is defined as the ratio between the largest possible utility that an agent can gain by manipulation and his utility in honest behavior.
The definition of incentive ratio is also closely related to \emph{approximate Nash equilibrium} \cite{DBLP:journals/sigecom/Rubinstein17, DBLP:journals/teco/CaragiannisFGS15}.
In particular, the result of \cite{DBLP:conf/icalp/ChenDZZ12} implies that the MNW mechanism achieves an incentive ratio of $2$ for homogeneous divisible items.\footnote{\nb{They show that Fisher market, the outcome of which coincides with the MNW allocation when agents have equal entitlements \cite[Volume 2, Chapter 14]{arrow1981handbook}, has an incentive ratio of $2$. Their result holds even for the more general setting in which agents have unequal entitlements.}}

As we have seen so far, there are flourishing positive results regarding Nash welfare against strategic agents for homogeneous divisible items.
Nevertheless, they either crucially rely on the assumption that agents are not capable of manipulating the boundaries between items \cite{DBLP:conf/sigecom/ColeGG13} or are established by the analysis on Fisher markets \cite{DBLP:journals/ior/BranzeiGM22, DBLP:conf/icalp/ChenDZZ12}, which bring up technical challenges to extend the positive results to cake cutting.
Therefore, a natural question arises: \emph{Do the celebrated MNW mechanism and the mechanism by \cite{DBLP:conf/sigecom/ColeGG13} provide any incentive guarantee in cake cutting?}

\subsection{Our Results}

In this paper, we answer the above question affirmatively.
Our main contribution is showing that for piecewise constant valuations, the incentive ratio of the MNW mechanism is $2$ in cake cutting.
Remarkably, it holds even without the free disposal assumption, which turns out to be critical in the design of truthful cake cutting mechanisms and is notoriously hard to get rid of~\cite{DBLP:journals/geb/ChenLPP13,DBLP:journals/scw/BeiHS20}.

Moreover, we show that the Partial Allocation (PA) mechanism proposed by \cite{DBLP:conf/sigecom/ColeGG13}, which is truthful and $1 / e$-MNW for homogeneous divisible items, has an incentive ratio between $[e^{1 / e}, e]$ in cake cutting.
If randomization is allowed, then the PA mechanism can be turned to be truthful in expectation.

Noticing that the MNW mechanism and the PA mechanism provide two alternatives for a trade-off between incentive ratio and Nash welfare, we interpolate between them by properly adjusting the parameters in the PA mechanism.
In particular, for cake cutting, we show that for every $c \in [0, 1]$, there exists a mechanism with incentive ratio $2^{1 - c}$ and $e^{-c}$-MNW guarantee and the mechanism is deterministic if and only if $c = 0$.
For homogeneous divisible items, we establish the same trade-off curve with the mechanism deterministic for all $c \in [0, 1]$.
Then we give a lower bound for the interpolation by showing that there does not exist a deterministic mechanism for homogeneous divisible items with an incentive ratio at most $2^{1 - c}$ that can guarantee an approximation ratio for MNW strictly greater than $2^{-c}$.

Finally, we study the optimal incentive ratio achievable by envy-free cake cutting mechanisms.
We first give an envy-free mechanism for two agents with an incentive ratio of $4 / 3$, largely improving the incentive ratio of $2$ provided by the MNW mechanism.
Then, we show that any envy-free cake cutting mechanism has an incentive ratio of $\Theta(n)$ when each bundle must be a connected piece.

\subsection{Related Work}

\nb{In this paper, we assume each agent's value density function to be piecewise constant and consider \emph{direct} mechanisms, i.e., all agents directly report their preferences to the mechanism.
Moreover, we require all agents to report simultaneously.}
Another popular query model to communicate with agents is called \emph{the Robertson-Webb query model} \cite{DBLP:books/daglib/0017738}, where agents reveal their valuations to the mechanism through iterative queries and, from the game-theoretic point of view, the game that agents are playing is an \emph{extensive-form game}, creating more room for agents to manipulate.
Under the Robertson-Webb query model, \cite{DBLP:conf/aaai/KurokawaLP13} shows that there does not exist a truthful and envy-free mechanism that terminates within a bounded number of queries, and \cite{DBLP:conf/ijcai/BranzeiM15} further proves that for any truthful mechanism, there always exists an agent receiving a value zero.

\nb{For homogeneous divisible items, truthfulness is more tractable, and one can hope for simultaneously achieving truthfulness and other objectives.}
A series of research focuses on the best approximation ratio for \emph{social welfare} achievable by truthful mechanisms \cite{DBLP:conf/atal/GuoC10,DBLP:conf/wine/HanSTZ11,DBLP:conf/atal/ColeGG13,DBLP:conf/ijcai/Cheung16}.
\cite{DBLP:conf/iat/ZivanDOS10} aims to achieve envy-freeness and Pareto optimality while reducing, but not eliminating, agents' incentive to misreport.

The design of fair and truthful mechanisms for allocating indivisible items has also drawn wide attention.
\cite{DBLP:conf/sigecom/AmanatidisBCM17} characterizes truthful mechanisms for two agents with additive valuations and show that truthfulness is incompatible with any meaningful fairness notion.
To circumvent the negative result, many relaxed notions of truthfulness are proposed \cite{brams2006better,DBLP:conf/nips/0001V22,DBLP:journals/ai/BuST23}, and the restricted setting of dichotomous valuations is also studied \cite{DBLP:conf/aaai/BabaioffEF21,DBLP:conf/wine/0002PP020,DBLP:conf/aaai/BarmanV22}.
Another series of research focuses on the existence of \emph{pure Nash equilibria} and their fairness properties with respect to the underlying true valuations \cite{DBLP:journals/corr/abs-2301-13652,DBLP:conf/wine/AmanatidisBFLLR21}.

The MNW mechanism turns out to be central for allocating indivisible items as well, despite the NP-hardness for its computation \cite{DBLP:journals/amai/NguyenRR13}.
The MNW mechanism has been shown to provide promising fairness guarantees under various contexts, including \emph{envy-freeness up to one item (EF1)} \cite{DBLP:journals/teco/CaragiannisKMPS19}, \emph{envy-freeness up to any item (EFX)} when there are at most two possible values for items \cite{DBLP:journals/tcs/AmanatidisBFHV21}, and weak-weighted EF1 in the asymmetric case where agents have unequal entitlements \cite{DBLP:journals/teco/ChakrabortyISZ21}.
Note that the above results hold for additive valuations, and for the more restricted dichotomous additive valuations, \cite{DBLP:conf/wine/0002PP020} shows that the MNW mechanism is group strategyproof.

Another approach to bypass the negative result of incorporating truthfulness is to allow randomized mechanisms.
A simple mechanism given by \cite{DBLP:conf/sagt/MosselT10} is universal envy-free and truthful in expectation.
Moreover, \cite{DBLP:conf/wine/AzizY14} proposes a mechanism that is truthful in expectation, robust proportional, and satisfies \emph{unanimity}\footnote{\nb{A cake cutting mechanism satisfies unanimity if when each agent's most preferred $1/n$ length of the cake is disjoint from other agents, each agent receives his most preferred subset of the cake of length $1/n$.}}.
Moving to the allocation of indivisible items, \cite{garg2022efficient} shows that, even considering randomized mechanisms, every Pareto optimal and truthful mechanism is a serial dictatorship.

\nb{It is worthy of mentioning that \cite{DBLP:journals/ior/BranzeiGM22} shows that the \emph{price of anarchy (PoA)} of the MNW mechanism is at most $2$ for homogeneous divisible items.
In particular, the PoA of a mechanism here is defined as the largest ratio between the best Nash welfare achievable without strategic behaviors and the worst Nash welfare achieved at any equilibrium of the mechanism.
However, Brânzei \textit{et al.} establish the PoA of the MNW mechanism by utilizing the techniques from Fisher markets, which are tailored to homogeneous divisible items and are not applicable to cake cutting.}
\section{Preliminaries}

Suppose that there is a divisible and heterogeneous item, usually referred to as a cake, which is represented by interval $[0, 1]$.
Denote $N = [n]$ as the set of agents, and each agent has a \emph{density function} $f_i: [0, 1] \to \R_{\geq 0}$, which measures his preference on the cake.
In this paper, we assume density functions to be \emph{piecewise constant}; that is, we can partition the cake into \nb{finitely many} intervals $X_1, \ldots, X_m$ such that the value of each density function is a constant on each interval.
For simplicity, we write $f_i(I)$ as the value of $f_i$ on subset $I$ such that $I \subseteq X_t$ for some $X_t$.
The value of agent $i$ for a subset $S \subseteq [0, 1]$ is defined as $v_i(S) = \int_S f_i(x) dx$.
Note that the partition $\{X_1, \ldots, X_m\}$ is not pre-specified, but determined by the density functions reported by agents.
Without loss of generality, assume that each agent has a strictly positive value for the entire cake; \nb{otherwise, none of the properties would be violated by allocating an empty set to this agent, and hence we can simply remove him}.

An allocation is a tuple $A = (A_1, \ldots, A_n)$, where $A_i$ is the bundle allocated to agent $i$, such that $A_i$ is the union of finitely many intervals and $A_i \cap A_j$ has measure zero for $i \neq j$.
\nbr{Without loss of generality, we assume the intervals forming a bundle to be closed intervals with a positive length.}
Unless otherwise stated, we adopt the \emph{free disposal} assumption, which allows the mechanism to discard resources without incurring a cost, and thus $\bigcup_{i \in N} A_i = [0, 1]$ does not necessarily hold.\footnote{We refer the readers to \cite{DBLP:journals/geb/ChenLPP13} for a discussion why this assumption is crucial and hard to be removed in many applications.}
\nbr{In particular, the free disposal assumption is removed only in \Cref{sec:nash-welfare-piecewise-constant}, where we establish the incentive ratio of the MNW mechanism in cake cutting.}
When the free disposal assumption is made, we implicitly discard the subset $S$ that no agents value, i.e., $v_i(S) = 0$ for all $i \in N$.
We say that an allocation $A$ is \emph{envy-free} if for all agents $i, j$, we have $v_i(A_i) \geq v_i(A_j)$.
An allocation $A$ is called \emph{Pareto optimal} if there does not exist another allocation $A'$ such that $v_i(A'_i) \geq v_i(A_i)$ for all $i \in N$ with at least one inequality strict.

A (deterministic/randomized) cake cutting mechanism $M$ receives a tuple of density functions $\bdf = (f_1, \ldots, f_n)$ as input and (deterministically/randomly) outputs an allocation $(M_1(\bdf), \ldots, M_n(\bdf))$, where $M_i(\bdf)$ is the bundle received by agent $i$. 
Define the \emph{utility} of agent $i$ as $v_i(M_i(\bdf))$.
Now, we introduce the notion of incentive ratio, \nb{which is informally defined as the largest ratio between the (expected) utility of an agent under misreporting and that under truthful telling over all possible instances}.

\begin{definition}[Incentive Ratio]
    The \emph{incentive ratio} of a (randomized) mechanism $M$ is defined as
    \begin{align*}
        \sup_{n} \sup_{f_1, \ldots, f_n} \sup_{i \in N} \sup_{f_i'} \frac{\E[v_i(M_i(f_1, \ldots, f_i', \ldots, f_n))]}{\E[v_i(M_i(f_1, \ldots, f_i, \ldots, f_n))]},
    \end{align*}
    where the expectations are over the randomness of $M$.
\end{definition}

\nbr{Throughout this paper, we will adhere to the slightly unusual convention that $f_i'$ denotes the value density function misreported by agent $i$ instead of the derivative of $f_i$.}
Note that the incentive ratio of a mechanism is at least $1$ by simply setting $f_i' = f_i$.
We say that a mechanism is \emph{truthful} if its incentive ratio is exactly $1$, i.e., an agent cannot gain a higher utility by manipulating his density function.
\nbr{We remark here that when the mechanism is randomized, the truthfulness defined here is also known as \emph{truthfulness in expectation} or \emph{ex-ante truthfulness}.
For fixed $n$ and (randomized) mechanism $M$, we will also refer to the quantity
\begin{align*}
    \sup_{f_1, \ldots, f_n} \sup_{f_i'} \frac{\E[v_i(M_i(f_1, \ldots, f_i', \ldots, f_n))]}{\E[v_i(M_i(f_1, \ldots, f_i, \ldots, f_n))]}
\end{align*}
as the \emph{incentive ratio for agent $i$}, and we say that $M$ is \emph{truthful for agent $i$} if the incentive ratio for agent $i$ equals $1$.}

\paragraph{Homogeneous divisible items}
In the homogeneous divisible items setting, the set of intervals $\{X_1, \ldots, X_m\}$ is pre-specified by mechanisms before agents report their density functions.
Hence, instead of reporting a density function, it is sufficient for each agent to report $m$ valuations, specifying the value density on each of the intervals.
Another widely used description for the homogeneous divisible items setting is that there are $m$ homogeneous divisible items \cite{DBLP:conf/sigecom/ColeGG13,DBLP:conf/ijcai/Cheung16}, each item for an interval, and these two descriptions are equivalent.
From the first description, it is easy to see that, for homogeneous divisible items, where agents can only misreport their valuations for each interval, there are fewer rooms for agents to misreport than cake cutting, where the breakpoint positions of intervals are manipulable as well.
In this paper, we will use the second description.

\subsection{Maximum Nash Welfare}

Given an allocation $(A_1, \ldots, A_n)$, the \emph{Nash welfare} is defined as $(\prod_{i \in N} v_i(A_i))^{1/n}$.
In the Maximum Nash Welfare (MNW) mechanism, an allocation that maximizes Nash welfare is returned (if there are multiple such allocations, then it breaks the tie arbitrarily).
It is easy to see that the MNW solution is scale-invariant, i.e., scaling an agent's density function by an arbitrary positive factor does not affect the outcome.
We first give the resource monotonicity of the MNW mechanism.

\begin{theorem}[Resource Monotonicity~\cite{segal2019monotonicity}]\label{thm:resource-monotonicity}
    Let $I \subseteq [0, 1]$.
    Suppose $A^-$ is an MNW allocation on $[0, 1] - I$, and $A$ is an MNW allocation on $[0, 1]$.
    Then for each agent $i$, we have $v_i(A_i^-) \leq v_i(A_i)$.
\end{theorem}

By letting $I = \emptyset$, an important implication of \Cref{thm:resource-monotonicity} is that the values received by a particular agent in all MNW allocations are identical.

\begin{corollary}\label{coro:uniqueness-mnw}
    If $A$ and $A'$ are MNW allocations with respect to the same profile, then $v_i(A_i) = v_i(A_i')$ for all $i \in N$.
\end{corollary}

For $0 \leq \alpha \leq 1$, we say that an allocation $A$ is \emph{$\alpha$-approximate MNW} (\emph{$\alpha$-MNW} for short) if for an MNW allocation $A'$ and each agent $i \in N$, we have $v_i(A_i) \geq \alpha \cdot v_i(A_i')$.
By \Cref{coro:uniqueness-mnw}, the approximation factor for MNW is well-defined.

Notably, in an MNW allocation, the value received by each agent is strictly positive.
The following characterization of MNW allocations will be used repeatedly.

\begin{theorem}[\cite{DBLP:conf/ijcai/BeiCHTW17}]\label{thm:nash-welfare-condition}
    An allocation $A = (A_1, \ldots, A_n)$ is MNW if and only if the following condition holds for each interval $X_t$:
    \begin{align}\label{cond:nash-welfare}
        \text{if } A_i \cap X_t \neq \emptyset, \text{ then }
        \frac{f_i(X_t)}{v_i(A_i)} \geq \frac{f_j(X_t)}{v_j(A_j)}
        \text{ for all } i, j.
    \end{align}
\end{theorem}

To intuitively understand \Cref{thm:nash-welfare-condition}, if we add a subset of length $\epsilon$ with density $f_i(X_t)$ to $A_i$, then the Nash welfare will increase by a factor of
\begin{align*}
    \frac{v_i(A_i) + \epsilon \cdot f_i(X_t)}{v_i(A_i)}
    = 1 + \epsilon \cdot \frac{f_i(X_t)}{v_i(A_i)}.
\end{align*}
Thus, if agent $i$ receives a subset of $X_t$ with a positive length in an MNW allocation, then he must have the largest "increase rate" with respect to $X_t$, i.e., $\frac{f_i(X_t)}{v_i(A_i)}$.
With this explanation, given an allocation $A$, we say that agent $i$ \emph{deserves} interval $X_t$ if $\frac{f_i(X_t)}{v_i(A_i)} \geq \frac{f_j(X_t)}{v_j(A_j)}$ for all $j \neq i$.
\section{Incentive Guarantees of the MNW Mechanism}
\label{sec:nash-welfare-piecewise-constant}

In this section, we establish the incentive ratio of $2$ for the MNW mechanisms without the free disposal assumption.
Note that when the free disposal assumption is not made, the mechanism must allocate the cake completely.

\begin{theorem}\label{thm:incentive-ratio-of-MNW}
    The incentive ratio of the MNW mechanism without the free disposal assumption is $2$, and the lower bound holds even with the free disposal assumption.
\end{theorem}

Before proving \Cref{thm:incentive-ratio-of-MNW}, we first give a useful definition.

\begin{definition}\label{def:week-MNW-allocation}
    Given profile $(f_1, \ldots, f_n)$, we say that an allocation $(A_1, \ldots, A_n)$ is \emph{weakly MNW} if for each $i \in N$, we have $v_i(A_i) > 0$ and for each interval $X_t$,
    \begin{align*}
        \text{if } A_i \cap X_t \neq \emptyset, 
        \text{ then }
        \frac{f_i(X_t)}{v_i(A_i)} \geq \frac{f_j(X_t)}{v_j(A_j)} \text{ for } j = 2, \ldots, n.
    \end{align*}
\end{definition}

Note that the condition in \Cref{def:week-MNW-allocation} is slightly weaker than the condition in \Cref{thm:nash-welfare-condition}, and thus an MNW allocation is also weakly MNW, which implies the existence of weakly MNW allocations.
Intuitively, the relaxation can be described as, in a weakly MNW allocation, agent $1$ may not receive some parts of the cake even if he is the only agent that deserves them.
As a result, the subset received by agent $1$ in a weakly MNW allocation could be less valuable than that received in an MNW allocation.
Moreover, by \Cref{thm:nash-welfare-condition}, if $(A_1, \ldots, A_n)$ is weakly MNW with respect to $(f_1, \ldots, f_n)$, then $(A_2, \ldots, A_n)$ is MNW on $A_2 \cup \ldots \cup A_n$ with respect to $(f_2, \ldots, f_n)$.
Formally, we have the following lemma.

\begin{lemma}\label{lem:value-in-weakly-MNW-allocations}
    The maximum value that agent $1$ receives among all weakly MNW allocations equals to the value agent $1$ receives in an MNW allocation.
\end{lemma}

The proof of \Cref{lem:value-in-weakly-MNW-allocations} can be found in \ref{sec:proof-of-lem-value-in-weakly-MNW-allocations}, and we are ready to prove the upper bound in \Cref{thm:incentive-ratio-of-MNW}.

\begin{lemma}
    The incentive ratio of the MNW mechanism without the free disposal assumption is at most $2$.
\end{lemma}

\begin{proof}
    Due to symmetry, it is sufficient to prove the incentive ratio for agent $1$.
    Suppose that the density function of agent $i$ is $f_i$, and agent $1$ misreports his density function as $f_1'$.
    Let $(A_1, \ldots, A_n)$ and $(A_1', \ldots, A_n')$ be MNW allocations with respect to profiles $(f_1, \ldots, f_n)$ and $(f_1', f_2, \ldots, f_n)$, respectively.
    We hope to show that $v_1(A_1') \leq 2 v_1(A_1)$.
    Assume that $v_1(A_1') > 0$, and otherwise the statement trivially holds.
    From now on, let $X_1, \ldots, X_m$ be the partition of $[0, 1]$ induced by $f_1', f_1, \ldots, f_n$ such that the value of each of the density functions is a constant on each interval $X_t$.
    Since $(A_1', \ldots, A_n')$ is an MNW allocation with respect to $(f_1', f_2, \ldots, f_n)$, by \Cref{thm:nash-welfare-condition}, we have the following inequalities:
    \begin{align}
        &\frac{f_1'(X_t)}{v_1'(A_1')} \geq \frac{f_i(X_t)}{v_i(A_i')}, \qquad \text{for each } A_1' \cap X_t \neq \emptyset \text{ and } i = 2, \ldots, n, \text{ and} \label{eqn:weakly-MNW-property1}\\
        &\frac{f_i(X_t)}{v_i(A_i')} \geq \frac{f_j(X_t)}{v_j(A_j')}, \qquad \text{for each } A_i' \cap X_t \neq \emptyset \text{ and } i, j = 2, \ldots, n. \label{eqn:weakly-MNW-property}
    \end{align}

    The idea to prove the lemma is the following: we first find a subset $S \subseteq A_1'$ such that $v_1(S)$ is at least half of $v_1(A_1')$ and allocation $(S, A_2', \ldots, A_n')$ is weakly MNW on cake $S \cup A_2' \cup \ldots \cup A_n'$ with respect to profile $(f_1, \ldots, f_n)$.
    Then combined with \Cref{thm:resource-monotonicity} and \Cref{lem:value-in-weakly-MNW-allocations}, we have the desired conclusion.

    Let 
    \begin{align}\label{eqn:def-by-ratio}
        U = \left\{ x \in A_1' ~\Bigl\vert~ \frac{f_1(x)}{v_1(A_1')} \geq \frac{f_1'(x)}{2v_1'(A_1')} \right\}
    \end{align}
    and
    \begin{align*}
        \overline{U} = A_1' \setminus U = \left\{ x \in A_1' ~\Bigl\vert~ \frac{f_1(x)}{v_1(A_1')} < \frac{f_1'(x)}{2v_1'(A_1')} \right\}.
    \end{align*}
    Hence,
    \begin{align*}
        v_1(\overline{U})
        &= \int_{\overline{U}} f_1(x) dx
        < \int_{\overline{U}} \frac{f_1'(x) v_1(A_1')}{2 v_1'(A_1')} dx
        = \frac{v_1(A_1')}{2 v_1'(A_1')} \int_{\overline{U}} f_1'(x) dx\\
        &= \frac{v_1(A_1')}{2 v_1'(A_1')} v_1'(\overline{U})
        \leq \frac{v_1(A_1')}{2}.
    \end{align*}
    This implies that
    \begin{align*}
        v_1(U) = v_1(A_1') - v_1(\overline{U}) > \frac{v_1(A_1')}{2}.
    \end{align*}
    Since $v_1(U)$ is an integration, we can find a subset $S \subseteq U$ such that
    \begin{align}\label{eqn:half-value}
        v_1(S) = \frac{v_1(A_1')}{2}.
    \end{align}

    Consider the allocation $(S, A_2', \ldots, A_n')$ on the partial cake $S \cup A_2' \cup \ldots \cup A_n'$ for profile $(f_1', f_2, \ldots, f_n)$.
    For each interval $X_t$ such that $S \cap X_t \neq \emptyset$, we have
    \begin{align}\label{eqn:cond-for-partial-MNW}
        \frac{f_1(X_t)}{v_1(S)}
        = \frac{2 f_1(X_t)}{v_1(A_1')}
        \geq \frac{f_1'(X_t)}{v_1'(A_1')}
        \geq \frac{f_i(X_t)}{v_i(A_i')}
        \qquad \text{for } i = 2, \ldots, n,
    \end{align}
    where the inequalities from left to right follow from \eqref{eqn:half-value}, \eqref{eqn:def-by-ratio} and \eqref{eqn:weakly-MNW-property1}.
    Moreover, \eqref{eqn:weakly-MNW-property} still holds for $(S, A_2', \ldots, A_n')$ since $A_2', \ldots, A_n'$ remain unchanged.
    Coupled with \eqref{eqn:cond-for-partial-MNW}, we know that $(S, A_2', \ldots, A_n')$ is a weakly MNW allocation on $S \cup A_2' \cup \ldots \cup A_n'$ with respect to profile $(f_1, \ldots, f_n)$.

    By \Cref{thm:resource-monotonicity}, we know that $v_1(A_1)$ is no less than the value agent $1$ receives in an MNW allocation on the cake $S \cup A_2' \cup \ldots \cup A_n'$, which in turn, by \Cref{lem:value-in-weakly-MNW-allocations}, is no less than $v_1(S)$ as agent $1$ receives in the weakly MNW allocation.
    As a result, $v_1(A_1) \geq v_1(S) = \frac{v_1(A_1')}{2}$ by \eqref{eqn:half-value}, which yields our desired result.
\end{proof}

Next, we show that our upper bound is tight to finish the proof of \Cref{thm:incentive-ratio-of-MNW}.
In particular, we prove that the incentive ratio of the MNW mechanism is at least $2 - 1 / n$, which approaches $2$ as $n \to \infty$.
Notably, the hard instance provided by \cite{DBLP:conf/icalp/ChenDZZ12} only works for the more general setting in which agents have different \emph{endowments} and thus cannot be applied here.

\begin{lemma}\label{lem:lower-bound-for-MNW-mechanism}
    The incentive ratio of the MNW mechanism is at least $2 - 1 / n$, even with the free disposal assumption.
\end{lemma}

\begin{proof}
    For simplicity, we assume by scaling that the cake is represented by interval $[0, 3]$.
    Suppose that there are $n$ agents and their density functions are as follows:
    \begin{align*}
        f_1(x) =
        \begin{cases}
            n, & x \in [0, 1],\\
            n - 1, & x \in (1, 2],\\
            0, & x \in (2, 3],
        \end{cases}
        \quad \text{and} \quad
        f_i(x) = 
        \begin{cases}
            0, & x \in [0, 1],\\
            1, & x \in (1, 2],\\
            n - 1, & x \in (2, 3]
        \end{cases}
        \quad \text{for } i \geq 2.
    \end{align*}
    It is easy to verify using \Cref{thm:nash-welfare-condition} that an MNW allocation is the following: $A_1 = [0, 1]$, and $(1, 2]$ and $(2, 3]$ are evenly distributed among agents $2, \ldots, n$.

    Let $\epsilon > 0$ be a sufficiently small real number, and suppose that agent $1$ misreports his density function as
    \begin{align*}
        f_1'(x) =
        \begin{cases}
            \epsilon, & x \in [0, 1], \\
            1, & x \in (1, 2], \\
            n - 1, & x \in (2, 3].
        \end{cases}
    \end{align*}
    Note that an MNW allocation is the following: agent $1$ receives $[0, 2 - \frac{n-1}{n}\epsilon]$, and $(2 - \frac{n-1}{n}\epsilon, 2]$ and $(2, 3]$ are evenly distributed among agents $2, \ldots, n$.
    To verify using \Cref{thm:nash-welfare-condition} that this allocation, denoted as $(A_1', \ldots, A_n')$, is MNW with respect to $(f_1', f_2, \ldots, f_n)$, simply notice that
    \begin{align*}
        v_1'(A_1') = v_2(A_2') = \ldots = v_{n}(A_{n}') = 1 + \frac{\epsilon}{n}
    \end{align*}
    and the entire interval of $[0, 1]$ belongs to agent $1$.
    Therefore, the incentive ratio of the MNW mechanism is at least
    \begin{align*}
        \frac{v_1(A_1')}{v_1(A_1)}
        = \frac{2n - 1 - \frac{(n-1)^2}{n}\epsilon}{n},
    \end{align*}
    which approaches $2 - 1 / n$ as $\epsilon \to 0$.
\end{proof}
\section{The Partial Allocation Mechanism}

The Partial Allocation (PA) mechanism proposed by \cite{DBLP:conf/sigecom/ColeGG13} is truthful and $1 / e$-MNW for homogeneous divisible items.
In this section, we explore the incentive property of the PA mechanism in cake cutting.
Specifically, we show that the PA mechanism under the cake cutting setting has an incentive ratio between $e^{1/e}$ and $e$.

We start by describing the PA mechanism, which consists of three steps:
\begin{enumerate}
    \item Compute an MNW allocation $A$ based on the reported density functions.

    \item For each agent $i$, compute an MNW allocation $A^i = (A^i_1, \ldots, A^i_{i - 1}, A^i_{i + 1}, \ldots, A_n)$ that would arise in his absence.

    \item \label{item:step-3-of-PA-mechanism} For each interval $X_t$, allocate to each agent $i$ the rightmost fraction $y_i$ of $A_i \cap X_t$, where
    \begin{align*}
        y_i = \frac{\prod_{j \neq i}v_{j}(A_{j})}{\prod_{j \neq i} v_{j}(A^i_{j})}.
    \end{align*}
\end{enumerate}
We remark here that the chosen portion to be allocated in Step~\ref{item:step-3-of-PA-mechanism} can be arbitrary.

\begin{theorem}[\cite{DBLP:conf/sigecom/ColeGG13}]\label{thm:PA-mechanism-for-divisible-items}
    For homogeneous divisible items, the PA mechanism is truthful.
    Moreover, we have $1 / e \leq y_i \leq 1$, and the lower bound is tight.
\end{theorem}

The following lemma is an essential property of the PA mechanism, which is implicit in \cite{DBLP:conf/sigecom/ColeGG13} and will be used repeatedly.
For completeness, we also provide a proof.

\begin{lemma}\label{lem:property-for-PA-mechanism}
    For every agent $i$, profile $(f_1, \ldots, f_n)$ and density function $f_i'$, define $A$ and $y_i$ as the MNW allocation and parameter computed by the PA mechanism for profile $(f_1, \ldots, f_n)$, and define $A'$ and $y_i'$ analogously for profile $(f_1, \ldots, f_i', \ldots, f_n)$.
    Then we have $y_i \cdot v_i(A_i) \geq y_i' \cdot v_i(A_i')$.
\end{lemma}

\begin{proof}
    Since the allocation $A^i$ does not depend on the reported density function of agent $i$, the denominators of $y_i$ and $y_i'$ are the same.
    Thus, it suffices to prove that
    \begin{align*}
        v_i(A_i) \cdot \prod_{j \neq i} v_j(A_j) \geq v_i(A_i') \cdot \prod_{j \neq i} v_j(A_j'),
    \end{align*}
    which holds because $A$ is the allocation that maximizes $\prod_{i \in N} v_i(A_i)$.
\end{proof}

The result of $1 / e \leq y_i \leq 1$ in \Cref{thm:PA-mechanism-for-divisible-items} can be easily carried over to the cake cutting setting by noticing that, with the absence of agents' strategic behaviors, a cake cutting instance is equivalent to a homogeneous divisible items instance with each interval $X_t$ corresponding to an item.
For the incentive ratio of the PA mechanism in cake cutting, we show that it is between $[e^{1/e}, e]$.

\begin{theorem}\label{thm:ic-pa-cake}
    The incentive ratio of the PA mechanism in cake cutting is between $[e^{1/e}, e]$.
\end{theorem}

\begin{proof}
    We prove the upper and lower bounds separately.

    \paragraph{Upper bound}
    Due to symmetry, it is sufficient to prove the incentive ratio for agent $1$.
    Suppose that the density function of agent $i$ is $f_i$ and agent $1$ misreports his density function as $f_1'$.
    Define $A$ and $y_1$ as the MNW allocation and parameter computed by the PA mechanism for profile $(f_1, \ldots, f_n)$, and define $A'$ and $y_1'$ analogously for profile $(f_1', f_2, \ldots, f_n)$.
    By \Cref{lem:property-for-PA-mechanism}, we have
    \begin{align}\label{eqn:MNW-condition-in-PA-mechanism}
        y_1 \cdot v_1(A_1) \geq y_1' \cdot v_1(A_1').
    \end{align}
    Notice that when agent $1$ reports truthfully, his received value is exactly $y_1 \cdot v_1(A_1)$.
    On the other hand, when agent $1$ misreports, his received value with respect to his true valuation is at most $v_1(A_1')$.
    Therefore, the incentive ratio of the PA mechanism for agent $1$ is at most
    \begin{align*}
        \frac{v_1(A_1')}{y_1 \cdot v_1(A_1)}
        \leq \frac{1}{y_1'}
        \leq e,
    \end{align*}
    where the first inequality is by \eqref{eqn:MNW-condition-in-PA-mechanism} and the second inequality is because $y_1' \geq 1 / e$.

    \paragraph{Lower bound}
    We prove the lower bound by presenting an example.
    For convenience, assume by scaling that the cake is represented by $[0, n]$.
    Let $h = \left( \frac{n - 1}{n} \right)^{n - 1}$.
    The density function of agent $1$ is
    \begin{align*}
        f_1(x) = 
        \begin{cases}
            1, & x \in [1-h, 1],\\
            0, & \text{otherwise},
        \end{cases}
    \end{align*}
    and the density functions of agents $i = 2, \ldots, n$ satisfy $f_i(x) = 1$ for $x \in [0, n]$.
    Note that an MNW allocation for profile $(f_1, \ldots, f_n)$ is as follows: agent $1$ receives $[1 - h, 1]$, and the rest of the cake is evenly distributed among other agents.
    Moreover, an MNW allocation for profile $(f_2, \ldots, f_n)$ is obtained by evenly distributing the entire cake among agents $2, \ldots, n$.
    Thus, when all agents report truthfully, we have
    \begin{align*}
        y_1 = \frac{\left( \frac{n - h}{n - 1} \right)^{n - 1}}{\left( \frac{n}{n - 1} \right)^{n - 1}}
        = \left(1 - \frac{h}{n} \right)^{n - 1},
    \end{align*}
    and the value received by agent $1$ under the PA mechanism is $h \cdot y_1$.
    
    Now suppose that agent $1$ misreports his density function as
    \begin{align*}
        f_1'(x) =
        \begin{cases}
            1, & x \in [0, 1], \\
            0, & \text{otherwise}.
        \end{cases}
    \end{align*}
    Note that an MNW allocation for profile $(f_1', f_2, \ldots, f_n)$ is obtained by allocating each agent $i$ with interval $[i - 1, i]$.
    As a result,
    \begin{align*}
        y_1' = \frac{1}{\left( \frac{n}{n - 1} \right)^{n - 1}}
        = \left( \frac{n - 1}{n} \right)^{n - 1} = h.
    \end{align*}
    Thus, agent $1$ gets $[1 - h, 1]$ under the PA mechanism, and the value received by agent $1$ with respect to his true valuation is $h$.
    Therefore, the incentive ratio of the PA mechanism is at least
    \begin{align*}
        \frac{h}{h \cdot y_1}
        = \left( 1 - \frac{h}{n} \right)^{1 - n},
    \end{align*}
    which approaches $e^{1 / e}$ as $n \to \infty$.
\end{proof}

The proof of the lower bound in \Cref{thm:ic-pa-cake} reveals a substantial difference between the division of homogeneous divisible items and cake.
We have mentioned that, without strategic consideration, a cake cutting instance is equivalent to a homogeneous divisible items instance with each $X_t$ corresponding to an item.
However, with the strategic consideration, an agent can ``merge'' two items.
In the proof of the lower bound, agent $1$ merges the two items $[0,1-h]$ and $(1-h,1]$ by misreporting, and the mechanism only sees one item $[0,1]$.
If the PA mechanism always discards the $(1-y_i)$ fraction of the item ``from the left'', such a merging is beneficial as the left-hand side portion with zero value happens to be discarded.
This is why the PA mechanism, while being truthful in the homogeneous divisible items setting, has an incentive ratio of greater than $1$ (i.e., fails to be truthful) in the cake cutting setting.

\subsection{Randomization}
\label{sec:randomized-PA-mechanism}

Note that in the aforementioned realization of the PA mechanism, the discarded portion of each interval is deterministically chosen.
If randomization is allowed, then we show that by randomly discarding resources, the PA mechanism is truthful in expectation in cake cutting.
Specifically, in Step~\ref{item:step-3-of-PA-mechanism}, for each interval $X_t$ and each agent $i$, assume without loss of generality that $X_t \cap A_i$ is an interval, denoted as $[\ell, r]$.
To allocate a fraction $y_i$ of $[\ell, r]$ to agent $i$, we select a starting point $s \in [\ell, r]$ uniformly at random, and then allocate the subinterval starting from $s$ with length $y_i \cdot (r - \ell)$ to agent $i$ (treat the interval $[\ell, r]$ as a cycle).
We call the modified PA mechanism as the randomized PA mechanism.

\begin{theorem}\label{thm:randomized-PA-mechanism}
    The randomized PA mechanism is truthful in expectation.
\end{theorem}

\begin{proof}
    Due to symmetry, it is sufficient to prove the truthfulness for agent $1$.
    Suppose that the density function of agent $i$ is $f_i$ and agent $1$ misreports his density function as $f_1'$.
    Define $A$ and $y_1$ as the MNW allocation and parameter computed by the randomized PA mechanism for profile $(f_1, \ldots, f_n)$, and define $A'$ and $y_1'$ analogously for profile $(f_1', f_2, \ldots, f_n)$.
    By \Cref{lem:property-for-PA-mechanism}, we have
    \begin{align*}
        y_1 \cdot v_1(A_1) \geq y_1' \cdot v_1(A_1').
    \end{align*}
    Notice that when agent $1$ reports truthfully, his received value is exactly $y_1 \cdot v_1(A_1)$.
    On the other hand, when agent $1$ misreports, it is easy to see that the probability of each point in $A_1'$ to be allocated to agent $1$ is $y_1'$.
    Thus, the expected value received by agent $1$ is
    \begin{align*}
        \int_{A_1'} y_1' \cdot f_1(x) dx
        = y_1' \cdot v_1(A_1').
    \end{align*}
    Therefore, the incentive ratio of the randomized PA mechanism is
    \begin{align*}
        \frac{y_1' \cdot v_1(A_1')}{y_1 \cdot v_1(A_1)} \leq 1,
    \end{align*}
    which implies that it is truthful in expectation for agent $1$.
\end{proof}
\section{Interpolation}
\label{sec:interpolation}

Recall that the MNW mechanism and the PA mechanism have been proven to provide different trade-offs between incentive ratio and Nash welfare for both cake cutting and homogeneous divisible items.
In this section, our goal is to establish an interpolation between these two trade-offs by properly choosing the fractions of discarded resources in the (randomized) PA mechanism.

\subsection{Cake Cutting}

So far, in cake cutting, we have provided two alternatives to establish a trade-off between incentive ratio and Nash welfare.
In particular,
\begin{enumerate}
    \item the MNW mechanism has an incentive ratio of $2$ and is $1$-MNW, and
    \item the randomized PA mechanism has an incentive ratio of $1$ and is $1 / e$-MNW.
\end{enumerate}
In the following, we achieve an interpolation between these two trade-offs.

\begin{theorem}\label{thm:interpolation-for-cake}
    For every $c \in [0, 1]$, there exists a (randomized) mechanism in cake cutting with incentive ratio $2^{1 - c}$ and $e^{-c}$-MNW guarantee.
    In particular, the mechanism is deterministic if and only if $c = 0$.
\end{theorem}

Note that setting $c = 0$ and $c = 1$ respectively recovers the guarantee of the MNW mechanism and that of the randomized PA mechanism.

\begin{proof}
    Given $c \in [0, 1]$, the mechanism is identical to the randomized PA mechanism presented in \Cref{sec:randomized-PA-mechanism} except the fractions of discarded resources in Step~\ref{item:step-3-of-PA-mechanism} are defined as
    \begin{align*}
        y_i = \left( \frac{\prod_{j \neq i} v_j(A_j)}{\prod_{j \neq i} v_j(A_j^i)} \right)^c
    \end{align*}
    for all $i \in N$.
    Note that if $c = 0$, then the mechanism is identical to the MNW mechanism, which is deterministic, and otherwise it is randomized.
    The Nash welfare guarantee simply follows by noticing that $y_i \geq e^{-c}$ implied by Theorem~\ref{thm:PA-mechanism-for-divisible-items}.
    
    It remains to prove the incentive ratio.
    Due to symmetry, it is sufficient to prove the incentive ratio for agent $1$.
    Suppose that the density function of agent $i$ is $f_i$ and agent $1$ misreports his density function as $f_1'$.
    Define $A$ and $y_1$ as the MNW allocation and parameter computed by our mechanism for profile $(f_1, \ldots, f_n)$, and define $A'$ and $y_1'$ analogously for profile $(f_1', f_2, \ldots, f_n)$.
    By the same arguments as in the proof of \Cref{thm:randomized-PA-mechanism}, the utility of agent $1$ when reporting truthfully is $y_1 \cdot v_1(A_1)$ and his expected utility when manipulating is $y_1' \cdot v_1(A_1')$.
    Thus, it suffices to show that
    \begin{align*}
        2^{1 - c} \cdot y_1 \cdot v_1(A_1) \geq y_1' \cdot v_1(A_1').
    \end{align*}
    Since the denominators of $y_1$ and $y_1'$ are the same, it is equivalent to show that
    \begin{align*}
        2^{1 - c} \cdot v_1(A_1) \cdot \left( \prod_{j = 2}^n v_j(A_j) \right)^c  \geq v_1(A_1') \cdot \left( \prod_{j = 2}^n v_j(A_j') \right)^c.
    \end{align*}
    Note that $A$ is the allocation that maximizes $\prod_{i \in N} v_i(A_i)$, which implies that $\prod_{i \in N} v_i(A_i) \geq \prod_{i \in N} v_i(A_i')$.
    As a result, it is sufficient to prove that
    \begin{align*}
        2^{1 - c} \cdot [v_1(A_1)]^{1 - c} \geq [v_1(A_1')]^{1 - c},
    \end{align*}
    which holds since by \Cref{thm:incentive-ratio-of-MNW}, the incentive ratio of the MNW mechanism is at most $2$, i.e., $\frac{v_1(A_1')}{v_1(A_1)} \leq 2$.
\end{proof}

\subsection{Homogeneous Divisible Items}

In this subsection, we consider the interpolation for homogeneous divisible items.
We start by noticing that the upper bound for the incentive ratio of the MNW mechanism in cake cutting carries over to homogeneous divisible items since agents have fewer rooms to misreport with homogeneous divisible items.
Moreover, the instance given in the proof of \Cref{lem:lower-bound-for-MNW-mechanism} can be seen as an instance with three homogeneous divisible items, which implies that the lower bound also holds.

\begin{corollary}
    For homogeneous divisible items, the incentive ratio of the MNW mechanism without the free disposal assumption is $2$, and the lower bound holds even with the free disposal assumption.
\end{corollary}

Combined with \Cref{thm:PA-mechanism-for-divisible-items}, we have the following alternatives for a trade-off between incentive ratio and Nash welfare:
\begin{enumerate}
    \item the MNW mechanism has an incentive ratio of $2$ and is $1$-MNW, and
    \item the PA mechanism has an incentive ratio of $1$ and is $1 / e$-MNW.
\end{enumerate}
Note that compared to the cake cutting setting, the PA mechanism here is deterministic instead of randomized.
By applying the same technique as in the proof of \Cref{thm:interpolation-for-cake}, we establish a counterpart of \Cref{thm:interpolation-for-cake} for homogeneous divisible items.

\begin{theorem}
    For every $c \in [0, 1]$, there exists a deterministic mechanism for homogeneous divisible items with incentive ratio $2^{1 - c}$ and $e^{-c}$-MNW guarantee.
\end{theorem}

To give a lower bound for the interpolation, we apply the instances given by \cite{DBLP:conf/sigecom/ColeGG13}, which are used to show that no truthful mechanism for homogeneous divisible items can guarantee an approximation ratio for MNW greater than $1/2$.
Note that the following lower bound holds for homogeneous divisible items, in which case agents are more restrictive, and thus it also holds for cake cutting.
The proof of \Cref{thm:interp-lb} can be found in \ref{sec:proof-of-interp-lb}.

\begin{theorem}\label{thm:interp-lb}
    There does not exist a deterministic mechanism for homogeneous divisible items with an incentive ratio at most $2^{1 - c}$ that can guarantee an approximation ratio for MNW strictly greater than $2^{-c}$.
\end{theorem}
\section{Envy-free Mechanisms}

In this section, we study the optimal incentive ratio achievable by envy-free cake cutting mechanisms.
Firstly, we give an envy-free mechanism for two agents with an incentive ratio of $4 / 3$, surpassing the incentive ratio upper bound provided by the MNW mechanism, which is at least $3 / 2$ for two agents by \Cref{lem:lower-bound-for-MNW-mechanism}.
Then, we show that the incentive ratio of any envy-free cake cutting mechanism is $\Theta(n)$ when the bundle allocated to each agent must be a connected piece.

\subsection{Two Agents}
\label{sec:2-agents}

We first give an envy-free cake cutting mechanism for two agents with an incentive ratio of $4/3$, which works as follows.
\begin{enumerate}
    \item \label{item:step1-2-agent} Let $x, y \in [0, 1]$ be two points of the cake satisfying $v_1([0, x]) = v_1([x, 1])$ and $v_2([0, y]) = v_2([y, 1])$, and assume without loss of generality that $x \leq y$.
    \item \label{item:step2-2-agent} Let $p, q \in [x, y]$ be two points satisfying $v_1([x, p]) = v_1[p, y]$ and $v_2([x, q]) = v_2([q, y])$.
    \item If $p \leq q$, then let $A_1 = [0, p]$ and $A_2 = [p, 1]$; otherwise, let $A_1 = [0, x] \cup [p, y]$ and $A_2 = [x, p] \cup [y, 1]$.
\end{enumerate}
To see that this mechanism is envy-free, notice that $[0, x] \in A_1$, and hence $v_1(A_1) \geq v_1([0, x]) = v_1([0, 1]) / 2$.
Similarly, since $[y, 1] \in A_2$, we have $v_2(A_2) \geq v_2([y, 1]) = v_2([0, 1]) / 2$.

\begin{theorem}
    The incentive ratio of the above mechanism is $4/3$.
\end{theorem}

\begin{proof}
    We prove the lower and upper bounds separately.

    \paragraph{Lower bound}
    Let the density functions of agents be
    \begin{align*}
        f_1(x) = 
        \begin{cases}
            1, & x \in [0, 1/2],\\
            0, & \text{otherwise}
        \end{cases}
    \end{align*}
    and $f_2(x) = 1$ for $x \in [0, 1]$.
    When both agents report truthfully, the parameters computed in Step~\ref{item:step1-2-agent} are $x = 1 / 4$ and $y = 1 / 2$, and the parameters computed in Step~\ref{item:step2-2-agent} satisfy $p = q = 3 / 8$.
    Since $p \leq q$, the resulting allocation is $([0, 3 / 8], [3 / 8, 1])$, in which the utility of agent $1$ is $3 / 8$.

    Assume that agent $1$ misreports his density function to be $f_1'(x) = 1$ for $x \in [0, 1]$.
    In this case, the parameters computed in Step~\ref{item:step1-2-agent} satisfy $x = y = 1 / 2$, and the parameters computed in Step~\ref{item:step2-2-agent} are $p = q = 1 / 2$.
    Since $p \leq q$, the resulting allocation is $([0, 1 / 2], [1 / 2, 1])$, in which the utility of agent $1$ with respect to $f_1$ is $1 / 2$.
    Therefore, the incentive ratio of the mechanism is at least $(1 / 2) / (3 / 8) = 4 / 3$.

    \paragraph{Upper bound}
    We only prove the incentive ratio for agent $1$, and the incentive ratio for agent $2$ can be proved similarly.
    Let $x^T, y^T, p^T, q^T$ be the computed parameters when both agents report truthfully, and let $x^M, y^M, p^M, q^M$ be the computed parameters after agent $1$ misreports his density function.
    Note that $y^T = y^M$.
    Recall we assume without loss of generality that $x^T \leq y^T$.
    Let $A$ and $A'$ be the resulting allocations under the true and manipulated profiles, respectively.
    Since we always have $[0, x^T] \subseteq A_1$ and $v_1([x^T, p^T]) = v_1([p^T, y^T]) = v_1([x^T, y^T]) / 2$, we have
    \begin{align*}
        v_1(A_1) = v_1([0, x^T]) + \frac{v_1([x^T, y^T])}{2} = \frac{v_1([0, 1]) + v_1([x^T, y^T])}{2}.
    \end{align*}
    When $x^M > y^M$, since $[0, y^M] \subseteq A_2'$, it holds that
    \begin{align*}
        v_1(A_1')
        \leq v_1([y^M, 1]) = v_1([y^T,1])
        \leq v_1([x^T, 1])
        = \frac{v_1([0, 1])}{2}
        \leq v_1(A_1),
    \end{align*}
    in which case manipulation is not beneficial. On the other hand, when $x^M \leq y^M$, since $[y^M, 1] \subseteq A_2'$,
    \begin{align*}
        v_1(A_1')
        \leq v_1([0, y^M])
        = v_1([0, x^T]) + v_1([x^T, y^T])
        = \frac{v_1([0, 1])}{2} + v_1([x^T, y^T]).
    \end{align*}
    Therefore, the incentive ratio of the mechanism is at most
    \begin{align*}
        \frac{v_1(A_1')}{v_1(A_1)}
        = \frac{\frac{v_1([0, 1])}{2} + v_1([x^T, y^T])}{\frac{v_1([0, 1]) + v_1([x^T, y^T])}{2}}
        \leq \frac{4}{3},
    \end{align*}
    where the inequality is because $v_1([x^T, y^T]) \leq v_1([x^T, 1]) = v_1([0, 1]) / 2$.
\end{proof}

\subsection{Connected Pieces Setting}

In this subsection, we consider a special setting where the bundle allocated to each agent must be a connected piece. That is, only $n - 1$ cuts are allowed to split the cake for $n$ agents.
In this setting, we show that the incentive ratio of any envy-free cake cutting mechanism is $\Theta(n)$.

We remark here that the free disposal assumption is no longer reasonable under this setting as the connected pieces constraint would be violated.
However, the hard instance constructed in the proof of \Cref{theorem-n-lower-bound} does not contain an interval that no agents value, and thus \Cref{theorem-n-lower-bound} holds even with the free disposal assumption.

\begin{theorem}\label{theorem-n-lower-bound}
    The incentive ratio of any envy-free cake cutting mechanism with the connected pieces constraint is $\Theta(n)$.
\end{theorem}

The upper bound in \Cref{theorem-n-lower-bound} is trivial due to the observation that the incentive ratio of any envy-free mechanism, even without the connected pieces requirement, is at most $n$ as each agent can always receive at least $1 / n$ of his value for the entire cake under honest behaviors, and the proof of the lower bound can be found in \ref{sec:proof-of-theorem-n-lower-bound}.
\section{Discussion and Future Directions}

In this paper, we show that the incentive ratio of $2$ for the MNW mechanism, which has been established for homogeneous divisible items, also holds in cake cutting.
This result not only further substantiates the prominence of the MNW mechanism and enhances our understanding of the differences between the homogeneous divisible items setting and cake cutting, but more importantly, reveals the possibility to combine fairness and efficiency with a low incentive ratio in cake cutting.
Note that an incentive ratio larger than $1$ can also be seen as a relaxation of truthfulness.
Thus, the first interesting direction is to explore the fairness and efficiency guarantees achievable when a low incentive ratio is required.

We also study the performance of the PA mechanism in cake cutting and show that its incentive ratio is between $[e^{1/e}, e]$.
An open problem is to further tighten the bounds, and if the incentive ratio of the PA mechanism is less than $2$, then as we did in \Cref{sec:interpolation}, we can interpolate between the guarantees provided by the MNW mechanism and the PA mechanism in cake cutting, which is realized by deterministic mechanisms as opposed to randomized ones.
Notably, whether we can remove the randomization condition and obtain a truthful deterministic mechanism with a non-trivial Nash welfare guarantee in cake cutting still remains unknown.
In fact, as an open question raised by \cite{DBLP:journals/ai/BuST23}, we do not even know whether there is a truthful cake cutting mechanism that always allocates each agent a subset on which the agent has a strictly positive value.

Moreover, recall that our interpolation for cake cutting is randomized.
However, to the best of our knowledge, there is no corresponding lower bound when randomization is allowed.
Thus, it is interesting to study the best trade-off between incentives and efficiency achieved by randomized mechanisms.

In addition, although the assumption of piecewise constant valuations is common in the cake cutting literature \cite{DBLP:conf/wine/AzizY14, DBLP:journals/ai/BuST23}, it would be interesting to generalize our results beyond piecewise constant valuations, for which there are two technical barriers.
Firstly, \Cref{thm:nash-welfare-condition}, which we heavily rely on in the proof of \Cref{thm:incentive-ratio-of-MNW}, is established under the assumption of piecewise constant valuations.
Secondly, the PA mechanism is defined for homogeneous divisible items, which is mathematically equivalent to cake cutting with piecewise constant valuations when the partition of the cake is fixed.
Therefore, to generalize the results in this paper beyond piecewise constant valuations, the first step is to generalize \Cref{thm:nash-welfare-condition} and the PA mechanism, which seems to be non-trivial.

Finally, we give an envy-free cake cutting mechanism for two agents with an incentive ratio of $4 / 3$.
It is intriguing to explore whether this bound is optimal and how to generalize the idea to more than two agents.
On the other hand, for the connected pieces setting, we show that envy-freeness and a low incentive ratio are incompatible.
It would be appealing to study whether we can combine a low incentive ratio with other weaker fairness criteria, such as proportionality, under the connected pieces constraint.

\section*{Acknowledgments}

The research of Biaoshuai Tao was supported by the National Natural Science Foundation of China (Grant No. 62472271 
and 62102252).
The research of Xiaohui Bei was supported by the Ministry of Education, Singapore, under its Academic Research Fund Tier 1 (RG98/23).

The authors are grateful to Yixin Tao, Nick Gravin, and the anonymous reviewers for their valuable comments on this paper.

\bibliographystyle{alpha}
\bibliography{references}

\newcommand{\etalchar}[1]{$^{#1}$}
\begin{thebibliography}{ABF{\etalchar{+}}21b}

\bibitem[ABCM17]{DBLP:conf/sigecom/AmanatidisBCM17}
Georgios Amanatidis, Georgios Birmpas, George Christodoulou, and Evangelos Markakis.
\newblock Truthful allocation mechanisms without payments: Characterization and implications on fairness.
\newblock In {\em {EC}}, pages 545--562. {ACM}, 2017.

\bibitem[ABF{\etalchar{+}}21a]{DBLP:journals/tcs/AmanatidisBFHV21}
Georgios Amanatidis, Georgios Birmpas, Aris Filos{-}Ratsikas, Alexandros Hollender, and Alexandros~A. Voudouris.
\newblock Maximum nash welfare and other stories about {EFX}.
\newblock {\em Theor. Comput. Sci.}, 863:69--85, 2021.

\bibitem[ABF{\etalchar{+}}21b]{DBLP:conf/wine/AmanatidisBFLLR21}
Georgios Amanatidis, Georgios Birmpas, Federico Fusco, Philip Lazos, Stefano Leonardi, and Rebecca Reiffenh{\"{a}}user.
\newblock Allocating indivisible goods to strategic agents: Pure nash equilibria and fairness.
\newblock In {\em {WINE}}, volume 13112 of {\em Lecture Notes in Computer Science}. Springer, 2021.

\bibitem[ABL{\etalchar{+}}23]{DBLP:journals/corr/abs-2301-13652}
Georgios Amanatidis, Georgios Birmpas, Philip Lazos, Stefano Leonardi, and Rebecca Reiffenh{\"{a}}user.
\newblock Round-robin beyond additive agents: Existence and fairness of approximate equilibria.
\newblock {\em CoRR}, abs/2301.13652, 2023.

\bibitem[ADH13]{DBLP:conf/atal/AumannDH13}
Yonatan Aumann, Yair Dombb, and Avinatan Hassidim.
\newblock Computing socially-efficient cake divisions.
\newblock In {\em {AAMAS}}, pages 343--350. {IFAAMAS}, 2013.

\bibitem[AFG{\etalchar{+}}17]{DBLP:conf/aaai/AlijaniFGST17}
Reza Alijani, Majid Farhadi, Mohammad Ghodsi, Masoud Seddighin, and Ahmad~S. Tajik.
\newblock Envy-free mechanisms with minimum number of cuts.
\newblock In {\em {AAAI}}, pages 312--318. {AAAI} Press, 2017.

\bibitem[AIHS81]{arrow1981handbook}
Kenneth~Joseph Arrow, Michael~D Intriligator, Werner Hildenbrand, and Hugo Sonnenschein.
\newblock {\em Handbook of mathematical economics}, volume~1.
\newblock North-Holland Amsterdam, 1981.

\bibitem[AU20]{DBLP:conf/isaac/AsanoU20}
Takao Asano and Hiroyuki Umeda.
\newblock Cake cutting: An envy-free and truthful mechanism with a small number of cuts.
\newblock In {\em {ISAAC}}, volume 181 of {\em LIPIcs}, pages 15:1--15:16. Schloss Dagstuhl - Leibniz-Zentrum f{\"{u}}r Informatik, 2020.

\bibitem[AY14]{DBLP:conf/wine/AzizY14}
Haris Aziz and Chun Ye.
\newblock Cake cutting algorithms for piecewise constant and piecewise uniform valuations.
\newblock In {\em {WINE}}, volume 8877 of {\em Lecture Notes in Computer Science}, pages 1--14. Springer, 2014.

\bibitem[BCH{\etalchar{+}}17]{DBLP:conf/ijcai/BeiCHTW17}
Xiaohui Bei, Ning Chen, Guangda Huzhang, Biaoshuai Tao, and Jiajun Wu.
\newblock Cake cutting: Envy and truth.
\newblock In {\em {IJCAI}}, pages 3625--3631. ijcai.org, 2017.

\bibitem[BEF21]{DBLP:conf/aaai/BabaioffEF21}
Moshe Babaioff, Tomer Ezra, and Uriel Feige.
\newblock Fair and truthful mechanisms for dichotomous valuations.
\newblock In {\em {AAAI}}, pages 5119--5126. {AAAI} Press, 2021.

\bibitem[BGM22]{DBLP:journals/ior/BranzeiGM22}
Simina Br{\^{a}}nzei, Vasilis Gkatzelis, and Ruta Mehta.
\newblock Nash social welfare approximation for strategic agents.
\newblock {\em Oper. Res.}, 70(1):402--415, 2022.

\bibitem[BHS20]{DBLP:journals/scw/BeiHS20}
Xiaohui Bei, Guangda Huzhang, and Warut Suksompong.
\newblock Truthful fair division without free disposal.
\newblock {\em Soc. Choice Welf.}, 55(3):523--545, 2020.

\bibitem[BJK{\etalchar{+}}06]{brams2006better}
Steven~J Brams, Michael~A Jones, Christian Klamler, et~al.
\newblock Better ways to cut a cake.
\newblock {\em Notices of the AMS}, 53(11):1314--1321, 2006.

\bibitem[BJK07]{DBLP:conf/dagstuhl/BramsJK07a}
Steven~J. Brams, Michael~A. Jones, and Christian Klamler.
\newblock Better ways to cut a cake - revisited.
\newblock In {\em Fair Division}, volume 07261 of {\em Dagstuhl Seminar Proceedings}. Internationales Begegnungs- und Forschungszentrum fuer Informatik (IBFI), Schloss Dagstuhl, Germany, 2007.

\bibitem[BM15]{DBLP:conf/ijcai/BranzeiM15}
Simina Br{\^{a}}nzei and Peter~Bro Miltersen.
\newblock A dictatorship theorem for cake cutting.
\newblock In {\em {IJCAI}}, pages 482--488. {AAAI} Press, 2015.

\bibitem[BS23]{DBLP:conf/faw/BuS23}
Xiaolin Bu and Jiaxin Song.
\newblock Maximize egalitarian welfare for cake cutting.
\newblock In {\em {IJTCS-FAW}}, volume 13933 of {\em Lecture Notes in Computer Science}, pages 263--280. Springer, 2023.

\bibitem[BST23]{DBLP:journals/ai/BuST23}
Xiaolin Bu, Jiaxin Song, and Biaoshuai Tao.
\newblock On existence of truthful fair cake cutting mechanisms.
\newblock {\em Artif. Intell.}, 319:103904, 2023.

\bibitem[BT95]{brams1995envy}
Steven~J Brams and Alan~D Taylor.
\newblock An envy-free cake division protocol.
\newblock {\em The American Mathematical Monthly}, 102(1):9--18, 1995.

\bibitem[BT96]{DBLP:books/daglib/0017730}
Steven~J. Brams and Alan~D. Taylor.
\newblock {\em Fair division - from cake-cutting to dispute resolution}.
\newblock Cambridge University Press, 1996.

\bibitem[BTD92]{berliant1992fair}
Marcus Berliant, William Thomson, and Karl Dunz.
\newblock On the fair division of a heterogeneous commodity.
\newblock {\em Journal of Mathematical Economics}, 21(3):201--216, 1992.

\bibitem[BV22]{DBLP:conf/aaai/BarmanV22}
Siddharth Barman and Paritosh Verma.
\newblock Truthful and fair mechanisms for matroid-rank valuations.
\newblock In {\em {AAAI}}, pages 4801--4808. {AAAI} Press, 2022.

\bibitem[CCD{\etalchar{+}}19]{DBLP:journals/dam/ChenCDQY19}
Zhou Chen, Yukun Cheng, Xiaotie Deng, Qi~Qi, and Xiang Yan.
\newblock Agent incentives of strategic behavior in resource exchange.
\newblock {\em Discret. Appl. Math.}, 264:15--25, 2019.

\bibitem[CDL20]{DBLP:conf/ipps/ChengDL20}
Yukun Cheng, Xiaotie Deng, and Yuhao Li.
\newblock Tightening up the incentive ratio for resource sharing over the rings.
\newblock In {\em {IPDPS}}, pages 127--136. {IEEE}, 2020.

\bibitem[CDLY22]{DBLP:conf/sigecom/ChengDLY22}
Yukun Cheng, Xiaotie Deng, Yuhao Li, and Xiang Yan.
\newblock Tight incentive analysis on sybil attacks to market equilibrium of resource exchange over general networks.
\newblock In {\em {EC}}, pages 792--793. {ACM}, 2022.

\bibitem[CDT{\etalchar{+}}22]{DBLP:journals/iandc/ChenDTZZ22}
Ning Chen, Xiaotie Deng, Bo~Tang, Hongyang~R. Zhang, and Jie Zhang.
\newblock Incentive ratio: {A} game theoretical analysis of market equilibria.
\newblock {\em Inf. Comput.}, 285(Part):104875, 2022.

\bibitem[CDZ11]{DBLP:conf/esa/ChenDZ11}
Ning Chen, Xiaotie Deng, and Jie Zhang.
\newblock How profitable are strategic behaviors in a market?
\newblock In {\em {ESA}}, volume 6942 of {\em Lecture Notes in Computer Science}, pages 106--118. Springer, 2011.

\bibitem[CDZZ12]{DBLP:conf/icalp/ChenDZZ12}
Ning Chen, Xiaotie Deng, Hongyang Zhang, and Jie Zhang.
\newblock Incentive ratios of fisher markets.
\newblock In {\em {ICALP} {(2)}}, volume 7392 of {\em Lecture Notes in Computer Science}, pages 464--475. Springer, 2012.

\bibitem[CFGS15]{DBLP:journals/teco/CaragiannisFGS15}
Ioannis Caragiannis, Angelo Fanelli, Nick Gravin, and Alexander Skopalik.
\newblock Approximate pure nash equilibria in weighted congestion games: Existence, efficient computation, and structure.
\newblock {\em {ACM} Trans. Economics and Comput.}, 3(1):2:1--2:32, 2015.

\bibitem[CGG13a]{DBLP:conf/sigecom/ColeGG13}
Richard Cole, Vasilis Gkatzelis, and Gagan Goel.
\newblock Mechanism design for fair division: allocating divisible items without payments.
\newblock In {\em {EC}}, pages 251--268. {ACM}, 2013.

\bibitem[CGG13b]{DBLP:conf/atal/ColeGG13}
Richard Cole, Vasilis Gkatzelis, and Gagan Goel.
\newblock Positive results for mechanism design without money.
\newblock In {\em {AAMAS}}, pages 1165--1166. {IFAAMAS}, 2013.

\bibitem[Che16]{DBLP:conf/ijcai/Cheung16}
Yun~Kuen Cheung.
\newblock Better strategyproof mechanisms without payments or prior - an analytic approach.
\newblock In {\em {IJCAI}}, pages 194--200. {IJCAI/AAAI} Press, 2016.

\bibitem[CISZ21]{DBLP:journals/teco/ChakrabortyISZ21}
Mithun Chakraborty, Ayumi Igarashi, Warut Suksompong, and Yair Zick.
\newblock Weighted envy-freeness in indivisible item allocation.
\newblock {\em {ACM} Trans. Economics and Comput.}, 9(3):18:1--18:39, 2021.

\bibitem[CKM{\etalchar{+}}19]{DBLP:journals/teco/CaragiannisKMPS19}
Ioannis Caragiannis, David Kurokawa, Herv{\'{e}} Moulin, Ariel~D. Procaccia, Nisarg Shah, and Junxing Wang.
\newblock The unreasonable fairness of maximum nash welfare.
\newblock {\em {ACM} Trans. Economics and Comput.}, 7(3):12:1--12:32, 2019.

\bibitem[CLPP13]{DBLP:journals/geb/ChenLPP13}
Yiling Chen, John~K. Lai, David~C. Parkes, and Ariel~D. Procaccia.
\newblock Truth, justice, and cake cutting.
\newblock {\em Games Econ. Behav.}, 77(1):284--297, 2013.

\bibitem[EG59]{eisenberg1959consensus}
Edmund Eisenberg and David Gale.
\newblock Consensus of subjective probabilities: The pari-mutuel method.
\newblock {\em The Annals of Mathematical Statistics}, 30(1):165--168, 1959.

\bibitem[ES99]{edward1999rental}
Francis Edward~Su.
\newblock Rental harmony: Sperner's lemma in fair division.
\newblock {\em The American mathematical monthly}, 106(10):930--942, 1999.

\bibitem[GC10]{DBLP:conf/atal/GuoC10}
Mingyu Guo and Vincent Conitzer.
\newblock Strategy-proof allocation of multiple items between two agents without payments or priors.
\newblock In {\em {AAMAS}}, pages 881--888. {IFAAMAS}, 2010.

\bibitem[GP22]{garg2022efficient}
Rohan Garg and Alexandros Psomas.
\newblock Efficient mechanisms without money: Randomization won't let you escape from dictatorships, 2022.

\bibitem[HPPS20]{DBLP:conf/wine/0002PP020}
Daniel Halpern, Ariel~D. Procaccia, Alexandros Psomas, and Nisarg Shah.
\newblock Fair division with binary valuations: One rule to rule them all.
\newblock In {\em {WINE}}, volume 12495 of {\em Lecture Notes in Computer Science}, pages 370--383. Springer, 2020.

\bibitem[HSTZ11]{DBLP:conf/wine/HanSTZ11}
Li~Han, Chunzhi Su, Linpeng Tang, and Hongyang Zhang.
\newblock On strategy-proof allocation without payments or priors.
\newblock In {\em {WINE}}, volume 7090 of {\em Lecture Notes in Computer Science}, pages 182--193. Springer, 2011.

\bibitem[HWWZ24]{DBLP:journals/jcss/HuangWWZ24}
Haoqiang Huang, Zihe Wang, Zhide Wei, and Jie Zhang.
\newblock Bounded incentives in manipulating the probabilistic serial rule.
\newblock {\em J. Comput. Syst. Sci.}, 140:103491, 2024.

\bibitem[KLP13]{DBLP:conf/aaai/KurokawaLP13}
David Kurokawa, John~K. Lai, and Ariel~D. Procaccia.
\newblock How to cut a cake before the party ends.
\newblock In {\em {AAAI}}. {AAAI} Press, 2013.

\bibitem[LSX24]{DBLP:conf/atal/0037SX24}
Bo~Li, Ankang Sun, and Shiji Xing.
\newblock Bounding the incentive ratio of the probabilistic serial rule.
\newblock In {\em {AAMAS}}, pages 1128--1136. International Foundation for Autonomous Agents and Multiagent Systems / {ACM}, 2024.

\bibitem[Mou03]{DBLP:books/daglib/0017734}
Herv{\'{e}} Moulin.
\newblock {\em Fair division and collective welfare}.
\newblock {MIT} Press, 2003.

\bibitem[MT10]{DBLP:conf/sagt/MosselT10}
Elchanan Mossel and Omer Tamuz.
\newblock Truthful fair division.
\newblock In {\em {SAGT}}, volume 6386 of {\em Lecture Notes in Computer Science}, pages 288--299. Springer, 2010.

\bibitem[NRR13]{DBLP:journals/amai/NguyenRR13}
Trung~Thanh Nguyen, Magnus Roos, and J{\"{o}}rg Rothe.
\newblock A survey of approximability and inapproximability results for social welfare optimization in multiagent resource allocation.
\newblock {\em Ann. Math. Artif. Intell.}, 68(1-3):65--90, 2013.

\bibitem[OS22]{DBLP:journals/scw/OrtegaS22}
Josu{\'{e}} Ortega and Erel Segal{-}Halevi.
\newblock Obvious manipulations in cake-cutting.
\newblock {\em Soc. Choice Welf.}, 59(4):969--988, 2022.

\bibitem[Pro16]{DBLP:reference/choice/Procaccia16}
Ariel~D. Procaccia.
\newblock Cake cutting algorithms.
\newblock In {\em Handbook of Computational Social Choice}, pages 311--330. Cambridge University Press, 2016.

\bibitem[PV22]{DBLP:conf/nips/0001V22}
Alexandros Psomas and Paritosh Verma.
\newblock Fair and efficient allocations without obvious manipulations.
\newblock In {\em NeurIPS}, 2022.

\bibitem[Rub17]{DBLP:journals/sigecom/Rubinstein17}
Aviad Rubinstein.
\newblock Settling the complexity of computing approximate two-player nash equilibria.
\newblock {\em SIGecom Exch.}, 15(2):45--49, 2017.

\bibitem[RW98]{DBLP:books/daglib/0017738}
Jack~M. Robertson and William~A. Webb.
\newblock {\em Cake-cutting algorithms - be fair if you can}.
\newblock A {K} Peters, 1998.

\bibitem[SFG{\etalchar{+}}19]{DBLP:journals/algorithmica/SeddighinFGAT19}
Masoud Seddighin, Majid Farhadi, Mohammad Ghodsi, Reza Alijani, and Ahmad~S. Tajik.
\newblock Expand the shares together: Envy-free mechanisms with a small number of cuts.
\newblock {\em Algorithmica}, 81(4):1728--1755, 2019.

\bibitem[SHS19]{segal2019monotonicity}
Erel Segal-Halevi and Bal{\'a}zs~R Sziklai.
\newblock Monotonicity and competitive equilibrium in cake-cutting.
\newblock {\em Economic Theory}, 68(2):363--401, 2019.

\bibitem[Tod19]{todo2019analysis}
Taiki Todo.
\newblock Analysis of incentive ratio in top-trading-cycles algorithms.
\newblock In {\em Proceedings of the Annual Conference of JSAI 33rd (2019)}, pages 2F1E304--2F1E304. The Japanese Society for Artificial Intelligence, 2019.

\bibitem[TY23]{tao2023fair}
Biaoshuai Tao and Mingwei Yang.
\newblock Fair and almost truthful mechanisms for additive valuations and beyond.
\newblock {\em arXiv preprint arXiv:2306.15920}, 2023.

\bibitem[Wel85]{weller1985fair}
Dietrich Weller.
\newblock Fair division of a measurable space.
\newblock {\em Journal of Mathematical Economics}, 14(1):5--17, 1985.

\bibitem[XL20]{DBLP:conf/aaai/XiaoL20}
Mingyu Xiao and Jiaxing Ling.
\newblock Algorithms for manipulating sequential allocation.
\newblock In {\em {AAAI}}, pages 2302--2309. {AAAI} Press, 2020.

\bibitem[ZDOS10]{DBLP:conf/iat/ZivanDOS10}
Roie Zivan, Miroslav Dud{\'{\i}}k, Steven Okamoto, and Katia~P. Sycara.
\newblock Reducing untruthful manipulation in envy-free pareto optimal resource allocation.
\newblock In {\em {IAT}}, pages 391--398. {IEEE} Computer Society Press, 2010.

\end{thebibliography}

\clearpage
\appendix
\section{Proof of \Cref{lem:value-in-weakly-MNW-allocations}}
\label{sec:proof-of-lem-value-in-weakly-MNW-allocations}

For simplicity, we make the free disposal assumption in the proof, and it is easy to see that the lemma still holds after removing the free disposal assumption since adding dummy subsets does not affect agents' received values.
Before proving \Cref{lem:value-in-weakly-MNW-allocations}, we present some important properties for weakly MNW allocations.
We start by showing that weakly MNW allocations satisfy a slightly weaker resource monotonicity guarantee.

\begin{lemma}\label{lem:resource-monotonicity-for-weakly-MNW}
    Let $I \subseteq X_t$ be a subset of an interval $X_t$.
    Suppose that $A$ is a weakly MNW allocation on the cake $[0, 1] - I$.
    Then there exists a weakly MNW allocation $A^+ = (A_1^+, \ldots, A_n^+)$ on $[0, 1]$ such that $v_i(A_i) \leq v_i(A_i^+)$ for all $i \in N$.
\end{lemma}

\begin{proof}
    Note that $(A_2, \ldots, A_n)$ is an MNW allocation on $[0, 1] - I - A_1$.
    According to \Cref{thm:resource-monotonicity}, there exists an MNW allocation $(A_2^+, \ldots, A_n^+)$ for profile $(f_2, \ldots, f_n)$ on cake $[0, 1] - A_1$ such that $v_i(A_i) \leq v_i(A_i^+)$ for $i = 2, \ldots, n$.
    Let $A_1^+ = A_1$.
    It suffices to show that $A^+$ constructed above is weakly MNW.
    For each interval $X_t$, if $A_i \cap X_t \neq \emptyset$ and $i \geq 2$, then by the fact that $(A_2^+, \ldots, A_n^+)$ is MNW and \Cref{thm:nash-welfare-condition}, we have
    \begin{align*}
        \frac{f_i(X_t)}{v_i(A_i^+)} \geq \frac{f_j(X_t)}{v_j(A_j^+)}
    \end{align*}
    for $j = 2, \ldots, n$.
    If $A_1 \cap X_t \neq \emptyset$, then for $j = 2, \ldots, n$, it holds that
    \begin{align*}
        \frac{f_1(X_t)}{v_1(A_1^+)}
        = \frac{f_1(X_t)}{v_1(A_1)}
        \geq \frac{f_j(X_t)}{v_j(A_j)}
        \geq \frac{f_j(X_t)}{v_j(A_j^+)},
    \end{align*}
    where the first inequality is because $A$ is weakly MNW.
    This concludes the proof of \Cref{lem:resource-monotonicity-for-weakly-MNW}.
\end{proof}

The following lemma says that if we add a subset of resources that agent $1$ deserves to a weakly MNW allocation, then we can strictly increase agent $1$'s utility without violating the weak MNW property.

\begin{lemma}\label{lem:strict-resource-monotonicity-for-weak-MNW}
   Let $I \subseteq X_t$ be a subset of an interval $X_t$.
   Suppose that $A$ is a weakly MNW allocation on cake $[0, 1] - I$.
   If $\frac{f_1(I)}{v_1(A_1)} \geq \frac{f_j(I)}{v_j(A_j)}$ for all $j \in N$, then there exists a weakly MNW allocation $A^+$ on cake $[0, 1]$ such that $v_1(A_1) < v_1(A_1^+)$, and $v_i(A_i) \leq v_i(A_i^+)$ for $i = 2, \ldots, n$.
\end{lemma}

\begin{proof}
    We aim to construct a new weakly MNW allocation $A'$ on $([0, 1] - I) \cup \Upsilon$ for some subset $\Upsilon \subseteq I$ such that $v_1(A_1) < v_1(A_1')$ and $v_i(A_i) \leq v_i(A_i')$ for $i = 2, \ldots, n$, and then \Cref{lem:strict-resource-monotonicity-for-weak-MNW} is immediate by \Cref{lem:resource-monotonicity-for-weakly-MNW}.
    To this end, we start by allocating a small fraction of $I$ to each of the agents who deserve it such that the utility of each of them increases by the same factor.
    However, it may cause the issue that they no longer deserve some of their received subsets in $A$ due to the increase in utilities.
    We reconcile this by repeatedly "balancing" the utilities among agents that receive parts of the same interval in $A$ so that they can keep in $A'$ what they receive in $A$.

    We construct a directed graph $G$ with $n$ vertices representing $n$ agents, and we will use the terms agent and vertex interchangeably.
    There is an edge pointing from $i$ to $j$ whenever the following condition is satisfied for some interval $X_t$:
    \begin{align*}
        A_j \cap X_t \neq \emptyset \quad \text{and} \quad \frac{f_i(X_t)}{v_i(A_i)} = \frac{f_j(X_t)}{v_j(A_j)},
    \end{align*}
    and we use $S_{i,j}$ to denote this interval $X_t$ (if there are multiple such intervals, then choose an arbitrary one).
    In other words, edge $(i, j)$ exists if and only if there is an interval $X_t$ such that agent $i$ deserves $X_t$ and agent $j$ receives some of $X_t$.

    For an edge $(i, j)$ in $G$, if $f_i(S_{i, j}) = 0$, then $f_j(S_{i, j}) = 0$ by definition, which implies that $f_k(S_{i, j}) = 0$ for $k = 2, \ldots, n$ since $A$ is weakly MNW.
    Due to the free disposal assumption, we must have $f_1(S_{i, j}) > 0$.
    By transferring subset $A_j \cap S_{i, j}$ from $A_j$ to $A_1$, we obtain an allocation $A'$ satisfying $v_1(A_1') > v_1(A_1)$ and $v_k(A_k') = v_k(A_k)$ for $k = 2, \ldots, n$, and then \Cref{lem:strict-resource-monotonicity-for-weak-MNW} is immediate by \Cref{lem:resource-monotonicity-for-weakly-MNW}.
    From now on, assume that $f_i(S_{i, j}) > 0$ and $f_j(S_{i, j}) > 0$ for each edge $(i, j)$ in $G$.
    
    Define $R$ to be the set containing all agents $i$ such that $\frac{f_i(I)}{v_i(A_i)}$ is maximum.
    By the assumption that $\frac{f_1(I)}{v_1(A_1)} \geq \frac{f_j(I)}{v_j(A_j)}$ for all $j \in N$, we have $1 \in R$.
    We use the following procedure to construct a subgraph $T$ which is a forest:
    \begin{enumerate}
        \item Initialize $V(T) = R$ and $E(T) = \emptyset$.
        
        \item Whenever there is an edge $(i, j)$ such that $i \notin T$ and $j \in T$, add $i$ to $V(T)$ and $(i, j)$ to $E(T)$.
    \end{enumerate}
    It is easy to see that the second step terminates within $n$ steps, and thus $T$ is well-defined.
    For each $(i, j) \in E(T)$, let $i$ be a child of $j$.
    Thus, $T$ is a forest and $R$ is the set of roots in $T$.
    Define $\ell(i)$ as the size of the subtree rooted at $i$.

    Let $\epsilon$ be a sufficiently small real number.
    We modify the allocation $A$ to get $A'$ in the following way:
    \begin{enumerate}
        \item Each agent $i \in R$ takes a length of $\epsilon \ell(i) \frac{v_i(A_i)}{f_i(I)}$ from $I$.
        Denote $\Upsilon$ as the total portion on $I$ allocated to agents in $R$.

        \item For every edge $(i, j) \in E(T)$, transfer a legnth of $\epsilon \ell(i) \frac{v_i(A_i)}{f_i(S_{i,j})}$ in $A_j \cap S_{i, j}$ from $A_j$ to $A_i$.
    \end{enumerate}
    Note that the construction is valid since $\epsilon$ is made sufficiently small, and the resulting allocation $A'$ is on cake $([0, 1] - I) \cup \Upsilon$ where $\Upsilon \subseteq I$.
    Moreover, for each agent $j \in R$,
    \begin{align*}
        v_j(A_j')
        &= v_j(A_j) + f_j(I) \cdot \epsilon \ell(j) \frac{v_j(A_j)}{f_j(I)} - \sum_{(i, j) \in E(T)} f_j(S_{i, j}) \cdot \epsilon \ell(i) \frac{v_i(A_i)}{f_i(S_{i, j})}\\
        &= v_j(A_j) + \epsilon \ell(j) v_j(A_j) - \sum_{(i, j) \in E(T)} \epsilon \ell(i) v_j(A_j)\\
        &= v_j(A_j) + \epsilon \ell(j) v_j(A_j) - \epsilon v_j(A_j) (\ell(j) - 1)\\
        &= (1 + \epsilon) v_j(A_j),
    \end{align*}
    where the second equality is by the fact that $\frac{f_i(S_{i, j})}{v_i(A_i)} = \frac{f_j(S_{i, j})}{v_j(A_j)}$.
    Similarly, for each agent $j \in V(T) \setminus R$ with agent $k$ as his parent in $T$, 
    \begin{align*}
        v_j(A_j')
        = v_j(A_j) + f_j(S_{j, k}) \cdot \epsilon \ell(j) \frac{v_j(A_j)}{f_j(S_{j, k})} - \sum_{(i, j) \in E(T)} f_j(S_{i, j}) \cdot \epsilon \ell(i) \frac{v_i(A_i)}{f_i(S_{i, j})}
        = (1 + \epsilon) v_j(A_j).
    \end{align*}
    Finally, for each agent $j \notin V(T)$, $v_j(A_j') = v_j(A_j)$.
    Since $1 \in R$, we have $v_1(A_1') = (1 + \epsilon) v_1(A_1) > v_1(A_1)$.

    It remains to show that $A'$ is weakly MNW on $([0, 1] - I) \cup \Upsilon$.
    For agent $i$ and interval $X_t$ such that $A_i' \cap X_t \neq \emptyset$, we show that $\frac{f_i(X_t)}{v_i(A_i')} \geq \frac{f_j(X_t)}{v_j(A_j')}$ for $j = 2, \ldots, n$.
    \begin{itemize}
        \item If $i \notin V(T)$, then $A_i' = A_i$ and $A_i \cap X_t \neq \emptyset$.
        Since $A$ is weakly MNW and $v_j(A_j') \geq v_j(A_j)$ for all $j \in N$, we have $\frac{f_i(X_t)}{v_i(A_i')} = \frac{f_i(X_t)}{v_i(A_i)} \geq \frac{f_j(X_t)}{v_j(A_j)} \geq \frac{f_j(X_t)}{v_j(A_j')}$ for $j = 2, \ldots, n$.

        \item If $i \in V(T)$, then $v_i(A_i') = (1 + \epsilon) v_i(A_i)$.
        We consider two cases regarding whether $A_i \cap X_t$ is empty.

        If $A_i \cap X_t \neq \emptyset$, then $\frac{f_i(X_t)}{v_i(A_i)} \geq \frac{f_j(X_t)}{v_j(A_j)}$ for $j = 2, \ldots, n$ since $A$ is weakly MNW.
        For those $j$ such that $\frac{f_i(X_t)}{v_i(A_i)} > \frac{f_j(X_t)}{v_j(A_j)}$, we can make $\epsilon$ small enough so that $\frac{f_i(X_t)}{v_i(A_i')} = \frac{f_i(X_t)}{(1 + \epsilon)v_i(A_i)} > \frac{f_j(X_t)}{v_j(A_j)} \geq \frac{f_j(X_t)}{v_j(A_j')}$.
        For those $j$ such that $\frac{f_i(X_t)}{v_i(A_i)} = \frac{f_j(X_t)}{v_j(A_j)}$, there is an edge from $j$ to $i$ in $G$, implying $j \in V(T)$.
        Since $v_i(A_i') = (1 + \epsilon) v_i(A_i)$ and $v_j(A_j') = (1 + \epsilon) v_j(A_j)$, we have $\frac{f_i(X_t)}{v_i(A_i')} = \frac{f_j(X_t)}{v_j(A_j')}$ as well.

        If $A_i \cap X_t = \emptyset$, since we assume $A_i' \cap X_t \neq \emptyset$, agent $i$ must have acquired from his parent in $T$, denoted as agent $k$, some part of $X_t = S_{i, k}$ unless $i \in R$, in which case the weakly MNW condition clearly holds.
        As a result, $A_k \cap X_t \neq \emptyset$ and $\frac{f_i(X_t)}{v_i(A_i)} = \frac{f_k(X_t)}{v_k(A_k)}$ as edge $(i, k)$ exists.
        Since $A$ is weakly MNW, $\frac{f_i(X_t)}{v_i(A_i)} = \frac{f_k(X_t)}{v_k(A_k)} \geq \frac{f_j(X_t)}{v_j(A_j)}$ for $j = 2, \ldots, n$.
        Similarly, for those $j$ such that $\frac{f_i(X_t)}{v_i(A_i)} > \frac{f_j(X_t)}{v_j(A_j)}$, we can make $\epsilon$ small enough so that $\frac{f_i(X_t)}{v_i(A_i')} \geq \frac{f_j(X_t)}{v_j(A_j')}$.
        For those $j$ such that $\frac{f_i(X_t)}{v_i(A_i)} = \frac{f_j(X_t)}{v_j(A_j)}$, there is an edge from $j$ to $i$ in $G$, implying $j \in V(T)$.
        Thus, $\frac{f_i(X_t)}{v_i(A_i')} = \frac{f_j(X_t)}{v_j(A_j')}$.
    \end{itemize}
    Therefore, $A'$ is weakly MNW, concluding the proof of \Cref{lem:strict-resource-monotonicity-for-weak-MNW}.
\end{proof}

Now we are ready to prove \Cref{lem:value-in-weakly-MNW-allocations}, and the technique applied here is somehow analogous to the proof of \Cref{lem:strict-resource-monotonicity-for-weak-MNW}.
Assume for contradiction that $A = (A_1, \ldots, A_n)$ is a weakly MNW allocation such that $v_1(A_1)$ is the maximum over all weakly MNW allocations but $A$ is not MNW.
We will show that there exists another weakly MNW allocation $A' = (A_1', \ldots, A_n')$ such that $v_1(A_1') > v_1(A_1)$, which contradicts the maximality of $v_1(A_1)$.

We construct a directed graph $G$ as follows: the graph consists of $n$ vertices representing $n$ agents, and there is an edge from $i$ to $j$ if and only if on some interval $S_{i, j}$ (which is one of $\{X_1, \ldots, X_m\}$), we have
\begin{enumerate}
    \item If $i > 1$, then $\frac{f_i(S_{i, j})}{v_i(A_i)} = \frac{f_j(S_{i, j})}{v_j(A_j)}$ and $A_j \cap S_{i, j} \neq \emptyset$.

    \item If $i = 1$, then $\frac{f_i(S_{i, j})}{v_i(A_i)} > \frac{f_j(S_{i, j})}{v_j(A_j)}$ and $A_j \cap S_{i, j} \neq \emptyset$.
\end{enumerate}
As in the proof of \Cref{lem:strict-resource-monotonicity-for-weak-MNW}, we assume that $f_i(S_{i, j}) > 0$ and $f_j(S_{i, j}) > 0$ for each edge $(i, j)$ in $G$.
By our assumption that $A$ is not MNW, the outdegree of vertex $1$ must be at least one.
We consider two cases regarding whether vertex $1$ is contained in a cycle in $G$.

\paragraph{Case 1: vertex $1$ is in a cycle $C$}
Without loss of generality, assume the cycle to be $1 \to 2 \to \ldots \to k \to 1$.
Then we construct a new allocation $A' = (A_1', \ldots, A_n')$ as follows, with $\epsilon > 0$ sufficiently small:
\begin{enumerate}
    \item $A_i' = A_i$ for $i \geq k + 1$.
    \item For $i = 2, \ldots, k$, denote $j$ as the next vertex of $i$ in $C$, and transfer a length of $\frac{v_{j}(A_{j})}{f_{j}(S_{i, j})} \epsilon$ in $A_{j} \cap S_{i, j}$ from $A_{j}$ to $A_i$.
    \item Cut a subset of $A_2 \cap S_{1, 2}$ with length $\frac{v_2(A_2)}{f_2(S_{1, 2})}$ from $A_2$, and allocate a subset of it with length $\frac{v_1(A_1)}{f_1(S_{1, 2})} \epsilon$ to $A_1$.
    Leave the rest of the subset of length $\frac{v_2(A_2)}{f_2(S_{1, 2})} \epsilon - \frac{v_1(A_1)}{f_1(S_{1, 2})} \epsilon$, denoted by $I$, unallocated temporarily.
\end{enumerate}
It is easy to verify that $v_i(A_i') = v_i(A_i)$ for all $i \in N$ as in the cycle, each agent $i$ receives a portion with value $f_i(S_{i, j}) \frac{v_{j}(A_{j})}{f_{j}(S_{i, j})} \epsilon = \epsilon v_i(A_i)$ from the next agent $j$ and donates a portion with the same value to the last agent.
This implies that all existing inequalities in the weak MNW condition on all intervals still hold.

Since $A_2 \cap S_{1, 2} \neq \emptyset$ and $A$ is weakly MNW, we have $\frac{f_1(I)}{v_1(A_1')} = \frac{f_1(S_{1, 2})}{v_1(A_1)} > \frac{f_2(S_{1, 2})}{v_2(A_2)} \geq \frac{f_i(S_{1, 2})}{v_i(A_i)} = \frac{f_i(I)}{v_i(A_i')}$ for $i = 2, \ldots, n$.
By applying \Cref{lem:strict-resource-monotonicity-for-weak-MNW} to $A'$ and $I$, we know that there is a weakly MNW allocation $(A_1^+, \ldots, A_n^+)$ on the cake $[0, 1]$ such that $v_1(A_1^+) > v_1(A'_1) = v_1(A_1)$, leading to a contradiction.

\paragraph{Case 2: vertex $1$ is not in a cycle}
As the outdegree of vertex $1$ is at least one, assume without loss of generality that there is an edge from $1$ to $2$.
Let $V_1$ be the set of vertices from which vertex $1$ is reachable (including vertex $1$), and let $T_1$ be a spanning tree rooted at $1$ on the subgraph induced by $V_1$ with each edge pointing from the child to the parent.
Analogously, let $V_2$ be the set of vertices that are reachable from vertex $2$, and let $T_2$ be a spanning tree rooted at $2$ on the subgraph induced by $V_2$ with each edge pointing from the parent to the child.
Note that the edge directions in $T_1$ and $T_2$ are different.
Since vertex $1$ is not in a cycle, it holds that $V_1 \cap V_2 = \emptyset$.

Now we apply the technique used in the proof of \Cref{lem:strict-resource-monotonicity-for-weak-MNW} to take a sufficiently small portion $\Upsilon$ in $S_{1, 2} \cap A_2$ from $A_2$ and allocate it among agents in $T_1$ so that the value received by each agent in $T_1$ increases by the same factor.

To compensate agent $2$ for donating $\Upsilon$, we apply a reversed technique as in the proof of \Cref{lem:strict-resource-monotonicity-for-weak-MNW}.
Specifically, for each edge $(i, j)$ in $T_2$, we transfer a small portion in $A_j \cap S_{i, j}$ from $A_j$ to $A_i$.
Notice that the directions of all edges in $T_2$ are the reverse of those in $T$ in the proof of \Cref{lem:strict-resource-monotonicity-for-weak-MNW}.
Therefore, each agent in $T_2$ receives portions from his children (except for leaves who do not receive anything), and donates a portion to his parent (except for agent $2$ who directly donates $\Upsilon$ to agent $1$).
This is again the reverse of what we did in the proof of \Cref{lem:strict-resource-monotonicity-for-weak-MNW}, where each agent donates to his children and receives from his parent.
By a similar technique, we can ensure that the value received by each agent in $T_2$ decreases by the same factor.

After performing all the operations described above, denoting the resulting allocation as $(A_1', \ldots, A_n')$, there exist sufficiently small constants $\epsilon, \delta > 0$ such that $v_i(A_i') = (1 + \epsilon) v_i(A_i)$ for all $i \in V_1$, $v_i(A_i') = (1 - \delta) v_i(A_i)$ for all $i \in V_2$, and $v_i(A_i') = v_i(A_i)$ for all the remaining agents $i$.
Applying the same arguments in the proof of \Cref{lem:strict-resource-monotonicity-for-weak-MNW}, the weak MNW property can still be preserved for $A'$ given a small choice of $\Upsilon$.
Moreover, since agent $1$ is in $T_1$, we have $v_1(A_1') = (1 + \epsilon) v_1(A_1) > v_1(A_1)$, leading to a contradiction.
\section{Proof of \Cref{thm:interp-lb}}
\label{sec:proof-of-interp-lb}

    We will show that there is no mechanism with an incentive ratio at most $2^{1 - c}$ that can guarantee an approximation ratio for MNW greater than $2^{-c} \cdot \frac{n + 1}{n} + \epsilon$ for any constant $\epsilon > 0$ for all $n$-agent instances.
    Note that as $n \to \infty$ and $\epsilon \to 0$, the approximation ratio converges to $2^{-c}$ as promised by the theorem.
    Let $\rho = 2^{-c} \cdot \frac{n + 1}{n}$.
    Assume for contradiction that mechanism $M$ is $\left( \rho + \epsilon \right)$-MNW for some positive constant $\epsilon$ with incentive ratio $2^{1 - c}$ for all $n$-agent instances.
    We will construct $n + 1$ instances with $n$ agents and $m = (k + 1) n$ items, where $k > \frac{2}{\epsilon}$ is sufficiently large, and show that the incentive ratio and the approximation ratio for MNW cannot be simultaneously guaranteed in all these instances.

    In instance $i \leq n$, every agent $j \neq i$ has a valuation of $kn + 1$ for item $j$ and a valuation of $1$ for every other item; agent $i$ has a valuation of $1$ for all items.
    In other words, each agent except agent $i$ has a strong preference for one item, which is different for each one of them.
    In an MNW allocation for instance $i$, each agent $j \neq i$ receives the entire item $j$, and agent $i$ receives the remaining $nk + 1$ items.
    Denote those $nk + 1$ items received by agent $i$ in the MNW allocation as $G_i$.
    Since $M$ is $(\rho + \epsilon)$-MNW, agent $i$ must receive at least a total value of $(\rho + \epsilon) (kn + 1)$ under $M$, and thus the total value he receives in $G_i$ is at least
    \begin{align*}
        (\rho + \epsilon) (kn + 1) - (n - 1) \geq \left( \rho + \frac{\epsilon}{2} \right) (kn + 1)
    \end{align*}
    by the assumption that $k > \frac{2}{\epsilon}$.
    As a result, there is an item in $G_i$ such that at least $\left(\rho + \frac{\epsilon}{2}\right)$ fraction of it is allocated to agent $i$ under $M$.
    Without loss of generality, assume that it is item $i$.
    Thus, we have shown that, for instance $i \leq n$, mechanism $M$ will allocate to agent $i$ at least $\left(\rho + \frac{\epsilon}{2}\right)$ fraction of item $i$ and at least a total value of $\left( \rho + \frac{\epsilon}{2} \right)(kn + 1)$ in $G_i$.

    Now we define instance $n + 1$, in which every agent $i$ has a valuation of $nk + 1$ for item $i$ and a valuation of $1$ for all other items.
    In an MNW allocation for instance $n + 1$, each agent $i$ receives the entire item $i$ together with a total value of $k$ in items $n + 1, \ldots, (k + 1) n$.
    Note that in instance $n + 1$, every agent $i$ can unilaterally misreport his valuation to lead to instance $i$, and by doing this, his received value with respect to his valuation in instance $n + 1$ is at least
    \begin{align*}
        \left( \rho + \frac{\epsilon}{2} \right)(kn + 1) + \left( \rho + \frac{\epsilon}{2} \right) kn \geq \left( \rho + \frac{\epsilon}{2} \right) 2kn,
    \end{align*}
    where in the left hand side, the first term comes from the fraction of item $i$ that agent $i$ receives and the second term comes from the average fraction of the remaining items.
    Plugging in $\rho = 2^{-c} \cdot \frac{n + 1}{n}$ and by the incentive ratio of $2^{1 - c}$, the value received by agent $i$ in instance $n + 1$ under $M$ is at least
    \begin{align*}
        2^{c - 1} \cdot \left(2^{-c} \cdot \frac{n + 1}{n} + \frac{\epsilon}{2}\right) 2kn
        = kn + k + 2^{c - 1} \cdot \epsilon kn
        \geq kn + k + 2
    \end{align*}
    by the assumptions that $k > \frac{2}{\epsilon}$ and $n \geq 2$.
    However, the value received by each agent in an MNW allocation of instance $n + 1$ is $kn + k + 1$, contradicting the Pareto optimality of the MNW allocation.
\section{Proof of \Cref{theorem-n-lower-bound}}
\label{sec:proof-of-theorem-n-lower-bound}

We show this by providing an example. Our example consists of $n=8k-3$ agents, and we will use many different letters, $p,a,b,c$, to represent agents instead of just $a$, because there are too many agents in this example. To prove the validity of the example, we show that
\begin{enumerate}
  \item there are two agents $p_1$ and $p_2$ such that at least one of them can at most receive a value that is a little bit more than 1;
  \item both agents $p_1$ and $p_2$ can misreport their functions and guarantee a value almost $k$ (which is $\Theta(n)$).
\end{enumerate}

We use a ``toy" example with $k=3$ for illustration in this proof, and the analysis for general $k$ is hardly different.
Instead of $[0,1]$, let the cake be $[-105,105]$ which can be easily normalized to $[0,1]$.
We divide the agents into 5 groups, and in the following description, $f$ is zero on the unspecified intervals.
The density functions for all agents are shown in Figure \ref{EC1}.
\begin{itemize}
      \item Group 1: Agent $p_1,p_2$ shown in black color, which are the two agents mentioned earlier.
      \begin{eqnarray*}
        &&f_{p_1}(x)=1\mbox{ on }[-89,-88]\cup[-82,-81]\cup[-75,-74]\cup[97.6,98.7]\\
        &&f_{p_2}(x)=1\mbox{ on }[-98.7,-97.6]\cup[74,75]\cup[81,82]\cup[88,89]
      \end{eqnarray*}

      \item Group 2: Agent $a_{11},a_{12},b_{11},b_{12}$ shown in green color.
      \begin{eqnarray*}
        &&f_{a_{11}}(x)=1\mbox{ on }[-86.9,-85]\cup[-70,-69]\cup[-61,-60]\\
        &&f_{a_{12}}(x)=1\mbox{ on }[-79.9,-78]\cup[-34,-33]\cup[-25,-24]\\
        &&f_{b_{11}}(x)=1\mbox{ on }[2,3]\cup[11,12]\cup[76.1,78]\\
        &&f_{b_{12}}(x)=1\mbox{ on }[38,39]\cup[47,48]\cup[83.1,85]
      \end{eqnarray*}

      \item Group 3: Agent $a_{21},a_{22},b_{21},b_{22}$ shown in blue color.
      \begin{eqnarray*}
        &&f_{a_{21}}(x)=1\mbox{ on }[-85,-83.1]\cup[-48,-47]\cup[-39,-38]\\
        &&f_{a_{22}}(x)=1\mbox{ on }[-78,-76.1]\cup[-12,-11]\cup[-3,-2]\\
        &&f_{b_{21}}(x)=1\mbox{ on }[24,25]\cup[33,34]\cup[78,79.9]\\
        &&f_{b_{22}}(x)=1\mbox{ on }[60,61]\cup[69,70]\cup[85,86.9]
      \end{eqnarray*}

      \item Group 4: Agent $a_{31},a_{32},b_{31},b_{32}$ shown in red color.
      \begin{eqnarray*}
        &&f_{a_{31}}(x)=1\mbox{ on }[-103.3,-93]\cup[-70,-60]\cup[-59,-49]\cup[-48,-38]\\
        &&f_{a_{32}}(x)=1\mbox{ on }[-103.3,-93]\cup[-34,-24]\cup[-23,-13]\cup[-12,-2]\\
        &&f_{b_{31}}(x)=1\mbox{ on }[2,12]\cup[13,23]\cup[24,34]\cup[93,103.3]\\
        &&f_{b_{32}}(x)=1\mbox{ on }[38,48]\cup[49,59]\cup[60,70]\cup[93,103.3]
      \end{eqnarray*}

      \item Group 5: Agent $c_1,c_2,c_3,c_4,c_5,c_6,c_7$ shown in red dots. They are shown by dots because their valued intervals have a negligible length which is $0.2$.
      $$f_{c_1}(x)=1\mbox{ on }[-91.1,-90.9]\qquad f_{c_2}(x)=1\mbox{ on }[-72.1,-71.9]$$
      $$f_{c_3}(x)=1\mbox{ on }[-36.1,-35.9]\qquad f_{c_4}(x)=1\mbox{ on }[-0.1,0.1]\qquad f_{c_5}(x)=1\mbox{ on }[35.9,36.1]$$
      $$f_{c_6}(x)=1\mbox{ on }[71.9,72.1]\qquad f_{c_7}(x)=1\mbox{ on }[90.9,91.1]$$

\end{itemize}

Notice that all the density functions have value either 1 or 0 according to the above definition. To make the figure look nicer, the functions are drawn with slightly different heights, but they should represent the same value 1.

For a clearer description, we divide the cake into 6 areas (which are labelled by Roman numbers in Figure \ref{EC1}):
$$\mbox{Area I: }[-105,91]\qquad\mbox{Area II: }[-91,-72]\qquad \mbox{Area III: }[-72,0]$$
$$\mbox{Area IV: }[0,72]\qquad\mbox{Area V: }[72,91]\qquad \mbox{Area VI: }[91,105]$$
We can see that agents $c_1,c_2,c_4,c_6,c_7$ served as ``boundaries" of these areas.
We should always notice that all the agents must receive connected pieces in this Theorem. As a result, no agent can receive an allocation in more than one Areas. Otherwise one or more agents from Group 5 will receive no value at all.

\begin{figure}
    \centering
    \includegraphics{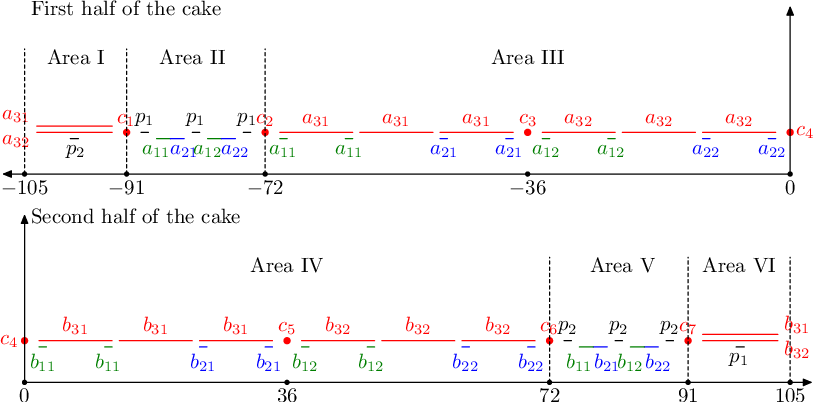}
    \caption{The Example with $k=3$}\label{EC1}
\end{figure}

To show how this example works, the following lemma is a crucial observation which is used many times in the proof. Notice that agent $p_1$ has three segments (with value 1 for each) in Area II and one segment (with value 1.1) in Area VI, and agent $p_2$ has three segments (with value 1 for each) in Area V and one segment (with value 1.1) in Area I which is symmetric to agent $p_1$.
\begin{lemma}\label{theta_n_lemma}
  If agent $p_1$ gets more than one segment in Area II, then agent $p_2$ has to take his valued interval in Area I which has value at most 1.1 (which means he cannot take any or any union of the three segments in Area V). Symmetrically, if agent $p_2$ gets more than one segment in Area V, then agent $p_1$ has to take his valued interval in Area VI.
\end{lemma}
\begin{proof}
  Let us see what happens if agent $p_1$ gets more than one segment in Area II. Without loss of generality, assume $p_1$ takes the first two intervals, which includes agents $(a_{11})'s,(a_{21})'s$ parts with value 1.9. Agents $a_{11},a_{21}$ then have to take from Area III, and out of total value 2, both of them will take at least 1.9 to keep envy-freeness. Now consider agent $a_{31}$. Each of agents $a_{11},a_{21}$ takes at least value 9.9 in $(a_{31})'s$ evaluation. Notice that $a_{31}$ cannot take cake in Area I, for otherwise one of $a_{11},a_{21}$ will take a value of at least 14.9 (the exact 14.9 happens when the middle segment of $a_{31}$ in Area III is shared evenly between $a_{31}$ and $a_{32}$) in his evaluation, which is more than all what he can get in Area I (which is 10.3). So $a_{31}$ can only get the middle segment in Area III. In particular, he can get at most value 10.2. Therefore, his valued segment with value 10.3 in Area I must be taken by at least two other agents to keep envy-freeness. There are only three candidates: $a_{32},p_2,c_1$, and the proof of lemma will finish after we show that agent $a_{32}$ cannot take cake in Area I (this then already implies $p_2$ has to take cake from Area I coupled with $c_1$).

  Assume negatively $a_{32}$ takes cake from Area I. He can take at most value 10.2 to avoid $a_{31}$ envying. In this case, his three valued segments in Area III must be shared by $c_3,c_4,a_{12},a_{22}$, and each of them must not get more than 10.2 by $(a_{32})'s$ evaluation. It can be easily checked that both $a_{12},a_{22}$ can get value strictly less than 1.9 by their own evaluation\footnote{If, say agent $a_{12}$, gets a value of 1.9 or more, then readers can try to construct such an allocation and easily realize that either the middle one of the three segments from $a_{32}$ must have some unallocated subset (which is not allowed as this subset is still valued by $a_{32}$), or agent $c_4$ has to take too much such that $a_{22}$ will envy him.}, and then their segments in Area II must be distributed to at least three other agents to maintain envy-freeness, which is clearly impossible due to the connected pieces constraint. So $a_{32}$ cannot take cake from Area I, and as discussed earlier, the lemma follows.
\end{proof}

Now, we show that both $p_1$ and $p_2$ can get almost all the three segments in Area II and Area V respectively, by misreporting their density functions. This will already imply the theorem:
\begin{enumerate}
  \item by the above lemma, at least one of $p_1,p_2$ will get value at most 1.1;
  \item this agent can lie to get value 3, which is $k$ in general.
\end{enumerate}

Without loss of generality, we see how agent $p_1$ can misreport. Let $\varepsilon$ be a very small positive real number.

$$f_{p_1}'(x)=\left\{
\begin{array}{ll}
  1 & \mbox{when }x\in[-89,-88]\cup[-82,-81]\cup[-75,-74]\\
  3-\varepsilon & \mbox{when }x\in[-105,-104]\cup[-93,-92]\\
  3-\varepsilon & \mbox{when }x\in[-71,-70]\cup[-60,-59]\cup[-49,-48]\cup[-38,-37]\\
  3-\varepsilon & \mbox{when }x\in[-35,-34]\cup[-24,-23]\cup[-13,-12]\cup[-2,-1]\\
  3-\varepsilon & \mbox{when }x\in[1,2]\cup[12,13]\cup[34,35]\\
  3-\varepsilon & \mbox{when }x\in[37,38]\cup[48,49]\cup[70,71]\\
  3-\varepsilon & \mbox{when }x\in[75,76]\cup[82,83]\cup[89,90]\\
  3-\varepsilon & \mbox{when }x\in[92,93]
\end{array}
\right.$$

\begin{figure}
    \centering
    \includegraphics{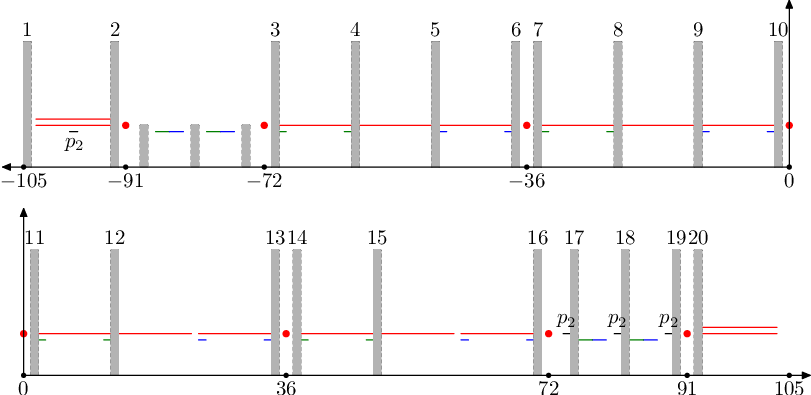}
    \caption{The Function that Agent $p_1$ Bids}\label{EC2}
\end{figure}

Figure \ref{EC2} illustrates how agent $p_1$ can lie. For the 4 segments he truly has value (three in Area II and one in Area VI), he only claims the three in Area II. Besides, he claims another 20 segments with height $3-\varepsilon$ ($k-\varepsilon$ in general) for each. As a result, $p_1$ now has 3 ``true" segments and 20 ``false" segments.

Now he has a total value of $3+20(3-\varepsilon)=21\times3-20\varepsilon$ on the whole cake, and there are 21 agents in this ``toy" example. We can see that by getting any of the 20 false segments which has value $3-\varepsilon$, it is not enough even for proportionality. Since the existence of agents $c_1$ and $c_2$ prevents $p_1$ getting a mixture of true and false segments, so $p_1$ will get a union of more than one of the 20 false segments, or almost all the 3 true segments. The latter case implies our claim already, so we only need to show the impossibility of the former case.

We label the 20 false segments as shown in Figure \ref{EC2}. We show pair by pair that $p_1$ cannot get any two adjacent segments of the twenty.

\begin{enumerate}
  \item 3 and 4 cannot be both taken, since more than half of $(a_{11})'s$ total value will be then taken by $p_1$, which is not envy-free. And the similar reason works for pairs 5 and 6, 7 and 8, 9 and 10, 11 and 12, 12 and 13, 14 and 15, 15 and 16.
  \item 6 and 7 cannot be both taken, since $(c_3)'s$ entire value is in between. This is also true for 2 and 3, 10 and 11, 13 and 14, 16 and 17, 19 and 20.
  \item 1 and 2 cannot be both taken. If agent $p_1$ takes 1 and 2 which worth $1.1$ by $(p_2)$'s evaluation, $p_2$ will have to take more than one segments in Area V. By Lemma \ref{theta_n_lemma}, $p_1$ must then take cake in Area VI, which contradicts to our assumption that $p_1$ takes 1 and 2.
  \item 4 and 5 cannot be both taken. Consider $a_{31}$. $p_1$ will occupy value 10 by taking 4 and 5 according to $a_{31}$. $a_{31}$ then must get a value at least 10 elsewhere. He cannot get it from Area III for otherwise he will take more than half of either $(a_{11})'s$ or $(a_{21})'s$ total value. If he take the one in Area I which worth $1.1$ according to $p_2$, then agent $p_2$ will take more than one segments in Area V, and by Lemma \ref{theta_n_lemma}, $p_1$ should take cake from Area VI, which contradicts to our assumption that $p_1$ takes 4 and 5 again.

      For the similar reason, $p_1$ cannot take pair 8 and 9.
  \item 17 and 18 cannot be both taken. If 17 and 18 are taken, the effect will be similar as if agent $p_2$ takes the first two segments in Area V. Following the analysis in the proof of Lemma \ref{theta_n_lemma}, we will also conclude that agent $p_1$ must take cake from Area VI, which is a contradiction again.

      For the similar reason, $p_1$ cannot take pair 18 and 19.
\end{enumerate}

Above enumerates all the possible pairs. Thus, $p_1$ has to take his true segments, and he needs to take almost 3 in order to maintain even proportionality. By misreporting, $p_1$ can get a value of almost $3$, and the same is for $p_2$. By the lemma, at least one of $p_1,p_2$ will get at most 1.1, so the theorem is proved for $k=3$.

For general $k$, $p_1$ and $p_2$ will have $k$ segments in Area II and Area V respectively. Correspondingly, instead of 4, there are $2(k-1)$ agents in each of Group 2, Group 3 and Group 4. They are
\begin{itemize}
  \item Group 2: $a_{11},a_{12},\ldots,a_{1(k-1)},b_{11},b_{12},\ldots,b_{1(k-1)}$;
  \item Group 3: $a_{21},a_{22},\ldots,a_{2(k-1)},b_{21},b_{22},\ldots,b_{2(k-1)}$;
  \item Group 4: $a_{31},a_{32},\ldots,a_{3(k-1)},b_{31},b_{32},\ldots,b_{3(k-1)}$.
\end{itemize}
It is easy to verify that Lemma \ref{theta_n_lemma} still holds for general $k$. As for agent $(p_1)$'s bid $f_{p_1}'(x)$, instead of $3-\varepsilon$ on 20 false segments, he will bid $k-\varepsilon$ on $8k-4$ segments at similar locations between the other agents. Then the rest of the analysis will be just the same.

\end{document}